\documentclass[journal, draftcls, onecolumn]{IEEEtran}

\usepackage{amsmath,amssymb,mathrsfs}
\usepackage{amsthm}
\usepackage{epsfig,epsf,subfigure,graphicx,graphics}
\usepackage{url}
\usepackage{hyperref}
\usepackage{array}
\usepackage{multicol}
\usepackage{empheq}
\usepackage{dsfont}
\usepackage{color}
\usepackage{bm}

\usepackage{multirow}
\usepackage{mathtools}
\usepackage{bbm}
\usepackage[algo2e,ruled]{algorithm2e}
\usepackage{verbatim}

\usepackage{bookmark}
\usepackage{hyperref}
\hypersetup{hidelinks=true}

\newtheoremstyle{mythmstyle} 
  {3pt}       
  {3pt}       
  {\itshape}  
  {}          
  {\bfseries} 
  {:}         
  {.5em}      
  {}          

\theoremstyle{mythmstyle}
\newcommand{\thistheoremname}{}
\newtheorem*{genericthm}{\thistheoremname}
\newenvironment{mythm}[1]{
  \renewcommand{\thistheoremname}{#1}
  \begin{genericthm}
}
{\end{genericthm}}

\theoremstyle{mythmstyle}
\newtheorem{lemma}{Lemma}[section]
\newtheorem{theorem}{Theorem}[section]

\newtheorem{corollary}{Corollary}[section]
\newtheorem{definition}{Definition}[section]
 \newtheorem{remark}{Remark}[section]

\newtheorem{example}{Example}[section]

\newtheorem{condition}{Condition}[section]

\renewcommand\thetable{\arabic{table}}

\newcommand{\lp}{\left(}
\newcommand{\rp}{\right)}
\newcommand{\lb}{\left[}
\newcommand{\rb}{\right]}
\newcommand{\lbp}{\left\{}
\newcommand{\rbp}{\right\}}
\newcommand{\lba}{\left\lvert}
\newcommand{\rba}{\right\rvert}
\newcommand{\lnorm}{\left\lVert}
\newcommand{\rnorm}{\right\rVert}

\newcommand{\mcal}{\mathcal}

\newcommand{\what}{\widehat}
\newcommand{\wtild}{\widetilde}
\newcommand{\mbf}{\mathbf}
\newcommand{\mbb}{\mathbb}
\newcommand{\msf}{\mathsf}
\newcommand{\mrm}{\mathrm}

\newcommand{\la}{\leftarrow}
\newcommand{\ra}{\rightarrow}

\newcommand{\da}{\downarrow}

\newcommand{\lfl}{\left\lfloor}
\newcommand{\rfl}{\right\rfloor}

\newcommand{\eqDef}{\triangleq}

\newcommand{\by}{\times}
\newcommand{\diid}{\overset{\text{i.i.d.}}{\sim}}

\newcommand{\st}{\text{ s.t. }}

\newcommand{\E}{\mathbb{E}}

\newcommand{\loss}{\ell}
\newcommand{\Risk}{\mathrm{R}}
\renewcommand{\Pr}{\mathbb{P}}
\newcommand{\Ber}{\mathrm{Ber}}
\newcommand{\Unif}{\mathrm{Unif}}

\newcommand{\Indc}[1]{\mathds{1}\lbp #1\rbp}
\newcommand{\circnum}[1]{\raisebox{.5pt}{\textcircled{\raisebox{-.9pt} {#1}}}}

\DeclareMathOperator*{\argmax}{arg\,max}
\DeclareMathOperator*{\argmin}{arg\,min}
\newcommand{\SBM}{$\mathtt{SBM}$}
\newcommand{\hSBM}{$\mathtt{hSBM}$}

\title{On the Minimax Misclassification Ratio of Hypergraph Community Detection}

\author{
I (Eli) Chien, Chung-Yi Lin, and I-Hsiang Wang
\thanks{I (Eli) Chien was with the Department of Electrical Engineering, National Taiwan University, Taipei 10617, Taiwan. He is now with the Department of Electrical and Computer Engineering, University of Illinois at Urbana-Champaign, Champaign, Illinois, USA (email:\url{ichien3@illinois.edu}).}
\thanks{
C.-Y. Lin is with the Graduate Institute of Communication Engineering, National Taiwan University, Taipei 10617, Taiwan (email: \url{r05942127@ntu.edu.tw}).}
\thanks{ 
I.-H. Wang is with the Department of Electrical Engineering and the Graduate Institute of Communication Engineering, National Taiwan University, Taipei 10617, Taiwan (email: \url{ihwang@ntu.edu.tw}).}
\thanks{The material in this paper was partly presented at the IEEE International Symposium on Information Theory, Aachen, Germany, June 2017 \cite{ISIT17_CDHG} and to partly appear in the International Conference on Artificial Intelligence and Statistics, Canary Islands, Spain, April 2018 \cite{AISTATS18_CDHG}.}
}
\begin{document}

\maketitle

\begin{abstract}

In this paper, community detection in hypergraphs is explored. Under a generative hypergraph model called ``{$d$-wise hypergraph stochastic block model}" ($d$-{\hSBM}) which naturally extends the Stochastic Block Model ({\SBM}) from graphs to $d$-uniform hypergraphs, the asymptotic minimax mismatch ratio (the risk function in the community detection problem) is characterized. For proving the achievability, we propose a two-step polynomial time algorithm that provably achieves the fundamental limit. The first step of the algorithm is a hypergraph spectral clustering method which achieves partial recovery to a certain precision level. The second step is a local refinement method which leverages the underlying probabilistic model along with parameter estimation from the outcome of the first step. 
To characterize the asymptotic performance of the proposed algorithm, we first derive a sufficient condition for attaining weak consistency in the hypergraph spectral clustering step. Then, under the guarantee of weak consistency in the first step, we upper bound the worst-case risk attained in the local refinement step by an exponentially decaying function of the size of the hypergraph and characterize the decaying rate. Compared to the existing works in {\SBM}, the main technical challenge lies in the complex structure of error events since community relations become much more complicated. 
We tackle the challenge by a series of non-trivial generalizations. To ensure the optimality of the refinement algorithm, all possible misclassification patterns among a wealth of community relations should be considered and treated with care when concentration inequalities are applied. Moreover, in proving the performance guarantee of the spectral clustering algorithm, we have to deal with random matrices whose entries are no longer independent as opposed to its graph counterpart. 
For proving the converse, the lower bound of the minimax mismatch ratio is set by finding a smaller parameter space which contains the most dominant error events, inspired by the analysis in the achievability part. 
It turns out that the minimax mismatch ratio decays exponentially fast to zero as the number of nodes tends to infinity, and the rate function is a weighted combination of several divergence terms, each of which is the R\'{e}nyi divergence of order $1/2$ between two Bernoulli distributions. The Bernoulli distributions involved in the characterization of the rate function are those governing the random instantiation of hyperedges in $d$-{\hSBM}. 
Experimental results on synthetic data validate our theoretical finding that the refinement step is critical in achieving the optimal statistical limit.

\end{abstract}

\section{Introduction}
\label{sec:intro}

Community detection (clustering) has received great attention recently across many applications, including social science, biology, computer science, and machine learning, while it is usually an ill-posed problem due to the lack of ground truth. A prevalent way to circumvent the difficulty is to formulate it as an inverse problem on a graph $\mcal{G} = \lbp \mcal{V},\mcal{E}\rbp$, where each node $i \in\mcal{V} = [n] \eqDef \{1,\ldots, n\}$ is assigned a community (label) $\sigma\lp i\rp \in [k] \eqDef \lbp1,\ldots,k\rbp$ that serves as the ground truth. 
The ground-truth \emph{community assignment}
$\sigma:[n] \ra [k]$
is hidden while the graph $\mcal{G}$ is revealed. Each edge in the graph models a certain kind of \emph{pairwise} interaction between the two nodes. The goal of community detection is to determine $\sigma$ from $\mcal{G}$, by leveraging the fact that
different combination of community relations leads to different likeliness of edge connectivity. 
When the graph $\mcal{G}$ is passively observed, community detection can be viewed as a statistical estimation problem, where the community assignment $\sigma$ is to be estimated from a statistical experiment governed by a generative model of random graphs. 
A canonical generative model is the \emph{stochastic block model} ({\SBM}) \cite{HollandLaskey_83} (also known as planted partition model \cite{CondonKarp_01}) which generates randomly connected edges from a set of labeled nodes. The presence of the $\binom{n}{2}$ edges is governed by $\binom{n}{2}$ independent Bernoulli random variables, and the parameter of each of them depends on the community assignments of the two nodes in the corresponding edge.

Through the lens of statistical decision theory, the fundamental statistical limits of community detection provides a way to benchmark various community detection algorithms. Under {\SBM}, the fundamental statistical limits have been characterized recently. One line of work takes a Bayesian perspective, where the unknown labeling $\sigma$ of nodes in $\mcal{V}$ is assumed to be distributed according to certain prior, and one of the most common assumption is i.i.d. over nodes. Along this line, the fundamental limit for exact recovery is characterized \cite{AbbeBandeira_16} in the full generality, while partial recovery remains open in general. See the survey \cite{Abbe_17} for more details and references therein. A second line of work takes a minimax perspective, and the goal is to characterize the minimax risk, which is typically the \emph{mismatch ratio} between the true community assignment and the recovered one. In \cite{ZhangZhou_16}, a tight asymptotic characterization of the minimax mismatch ratio for community detection in {\SBM} is found. 
Along with these theoretical results, several algorithms have been proposed to achieve these limits, including 
degree-profiling comparison \cite{AbbeSandon_15} for exact recovery, spectral MLE \cite{YunProutiere_16} for almost-exact recovery, and a two-step mechanism \cite{GaoMa_15} under the minimax framework.


However, graphs can only capture pairwise relational information, while such dyadic measure may be inadequate in many applications, such as the task of 3-D subspace clustering \cite{AgarwalLim_05} and the higher-order graph matching problem in computer vision \cite{DuchenneBach_11}. 
Moreover, in a co-authorship network such as the DBLP bibliography database where collaboration between scholars usually comes in a group-wise fashion, it seems more appropriate to represent the co-writing relationship in a single collective way rather than inscribing down each pairwise interaction \cite{ZhangHu_17}. 
Therefore, it is natural to model such \emph{beyond-pairwise} interaction by a hyperedge in a hypergraph and study the clustering problem in a hypergraph setting \cite{ZienSchlag_99}. \emph{Hypergraph partitioning} has been investigated in computer science, and several algorithms have been proposed, including spectral methods based on clique expansion \cite{AgarwalBranson_06}, hypergraph Laplacian \cite{ZhouHuang_06}, game-theoretic approaches \cite{BuloPelillo_13}, tensor method \cite{GhoshdastidarDukkipati_15}, linear programming \cite{LiDau_16}, to name a few. Existing approaches, however, mainly focus on optimizing a certain score function entirely based on the connectivity of the observed hypergraph and do not view it as a statistical estimation problem. 

In this paper, we investigate the community detection problem in hypergraphs through the lens of statistical decision theory. Our goal is to characterize the fundamental statistical limit and develop computationally feasible algorithms to achieve it. As for the generative model for hypergraphs, one natural extension of the {\SBM} model to a hypergraph setting is the \emph{hypergraph stochastic block model} ({\hSBM}), where the presence of an order-$h$ hyperedge $\msf{e}\subset\mcal{V}$ (i.e. $\lba\msf{e}\rba=h\leq M$, the maximum edge cardinality) is governed by a Bernoulli random variable with parameter $\theta_{\msf{e}}$ and the presence of different hyperedges are mutually independent. Despite the success of the aforementioned algorithms applied on many practical datasets, it remains open how they perform in {\hSBM} since the the fundamental limits have not been characterized and the probabilistic nature of {\hSBM} has not been fully utilized. 


As a first step towards characterizing the fundamental limit of community detection in hypergraphs, in this work we focus on the ``$d$-wise hypergraph stochastic block model" ($d$-{\hSBM}), where all the hyperedges generated in the hypergraph stochastic block model are of order $d$. 
Our main contributions are summarized as follows. 
\begin{itemize}
\item
First, we characterize of the asymptotic minimax mismatch ratio in $d$-{\hSBM} for any order $d$. 
\item 
Second, we propose a polynomial time algorithm which provably achieves the minimax mismatch ratio in the asymptotic regime, under mild regularity conditions. 
\end{itemize}
To the best of our knowledge, this is the first result which characterizes the fundamental limit on the minimax risk of community detection in random hypergraphs, together with a companion efficient recovery algorithm. 
The proposed algorithm consists of two steps. The first step is a global estimator that roughly recovers the hidden community assignment to a certain precision level, and the second step refines the estimated assignment based on the underlying probabilistic model. 



It is shown that the minimax mismatch ratio in $d$-{\hSBM} converges to zero exponentially fast as $n$, the size of the hypergraph, tends to infinity. The rate function, which is the exponent normalized by $n$, turns out to be a linear combination of R\'{e}nyi divergences of order 1/2. Each divergence term in the sum corresponds to a pair of community relations that would be confused with one another when there is only one misclassification, and the weighted coefficient associated with it indicates the total number of such confusing patterns. Probabilistically, there may well be two or more misclassifications, with each confusing relation pair pertaining to a R\'{e}nyi divergence when analyzing the error probability. However, we demonstrate technically that these situations are all dominated by the error event with a single misclassified node, which leaves out only the ``neighboring" divergence terms in the asymptotic expression. The main technical challenge resolved in this work is attributed to the fact that the community relations become much more complicated as the order $d$ increases, meaning that more error events may arise compared to the dichotomy situation (i.e. same-community and different-community) in the graph {\SBM} case. In the proof of achievability, we show that the second refinement step is able to achieve the fundamental limit provided that the first initialization step satisfies a certain weak consistency condition. The core of the second-step algorithm lies in a local version of the maximum likelihood estimation, where concentration inequalities are utilized to upper bound the probability of error. Here, an additional regularity condition is required to ensure that the probability parameters, which corresponds to the appearance of various types of hyperedge, do not deviate too much from each other. We would like to note that this constraint can be relaxed as long as the number of communities $k$ considered does not scale with $n$. For the first step, we use the tools from perturbation theory such as the Davis-Kahan theorem to prove the performance of the proposed spectral clustering algorithm. Since entries in the derived hypergraph Laplacian matrix are no longer independent, a union bound is applied here to make the analysis tractable with the concentration inequalities. The converse part of the minimax mismatch ratio follows a standard approach in statistics by finding a smaller parameter space where we can analyze the risk. We first lower bound the minimax risk by the Bayesian risk with a uniform prior. Then, the Bayesian risk is transformed to a local one by exploring the closed-under-permutation property of the targeted parameter space. Finally, we identify the local Bayesian risk with the risk function of a hypothesis testing problem and apply the Rozovsky lower bound in the large deviation theory to obtain the desired converse result.

\subsection*{Related Works}
The hypergraph stochastic block model is first introduced in \cite{GhoshdastidarDukkipati_14} as the planted partition model in random uniform hypergraphs where each hyperedge has the same cardinality. The uniform assumption is later relaxed in a follow-up work \cite{GhoshdastidarDukkipati_17} and a more general {\hSBM} with mixing edge-orders is considered. In \cite{AngeliniCaltagirone_15}, the authors consider the sparse regime and propose a spectral method based on a generalization of non-backtracking operator. Besides, a weak consistency condition is derived in \cite{GhoshdastidarDukkipati_17} for {\hSBM} by using the hypergraph Laplacian. Departing from {\SBM}, an extension to the \emph{censored block model} to the hypergraph setting is considered in \cite{AhnLee_16}, where an information theoretic limit on the sample complexity for exact recovery is characterized. 
As for the proposed two-step algorithm, the \emph{refine-after-initialize} concept has also been used in graph clustering \cite{AbbeSandon_15,YunProutiere_16,GaoMa_15} and ranking \cite{ChenSuh_15}.

This paper generalizes our previous works in the two conference papers \cite{ISIT17_CDHG,AISTATS18_CDHG} in three ways. First, \cite{ISIT17_CDHG} only explores the extension from the graph {\SBM} to $3$-{\hSBM} case where the observed hyperedges are $3$-uniform, as compared to a more general $d$-{\hSBM} model for any order $d$ analyzed in \cite{AISTATS18_CDHG} and here. In addition, the number of communities $k$ is allowed to be scaled with the number of vertices $n$ in this work, rather than being a constant as assumed in \cite{AISTATS18_CDHG}. This slight relaxation actually leads to another regularization condition imposed on the connecting probabilities, which is an non-trivial technical extension. Finally, we also demonstrate that our proposed algorithms: the hypergraph spectral clustering algorithm and the local refinement scheme are able to achieve the partial recovery and the exact recovery criteria, respectively. 

The rest of the paper is organized as follows.
We first introduce the random hypergraph model $d$-{\hSBM} and formulate the community detection problem in \autoref{sec:formu}. Previous efforts on the minimax result in graph {\SBM} and $3$-{\hSBM} are refreshed in \autoref{sec:moti}, which motivates the key quantities that characterize the fundamental limit in $d$-{\hSBM}. The main contribution of this work, the characterization of the optimal minimax mismatch ratio under $d$-{\hSBM} for any general $d$, is presented in \autoref{sec:main}. We propose two algorithms in \autoref{sec:algo} along with an analysis on the time complexity. Theoretical guarantees for the proposed algorithms as well as the technical proofs are given in \autoref{sec:theo}, while the converse part of the main theorem is established in \autoref{sec:conv}. {In \autoref{sec:expe}, we implement the proposed algorithms on synthetic data and present some experimental results. The paper is concluded with a few discussions on the extendability of the two-step algorithm and the minimax result to a weighted $d$-{\hSBM} setting in \autoref{sec:disc}.


\subsection*{Notations}
\begin{itemize}
\item Let $\lba S \rba$ denote the cardinality of the set $S$ and $[N] \eqDef \{ 1,2,\ldots,N \}$ for $N\in\mbb{N}$. 
\item $\mcal{S}_n$ is the symmetric group of degree $n$, which contains all the permutations from $[n]$ to itself.
\item The function $\sigma(\bm{x})\eqDef(\sigma(x_1),\ldots,\sigma(x_n))$ represents the community label vector associated with a labeling function $\sigma:[n]\ra[k]$ for a node vector $\bm{x}=(x_1,\ldots,x_n)$. 
\item For any community assignment $\bm{y}=(y_1,\ldots,y_n)\in[k]^n$ and permutation $\pi\in\mcal{S}_k$, $\bm{y}_{\pi}\eqDef(y_{\pi(1)},\ldots,y_{\pi(n)})$ denotes the permuted assignment vector.
\item The asymptotic equality between two functions $f(n)$ and $g(n)$, denoted as $f\asymp g$ (as $n\ra\infty$), holds if $\lim_{n\ra\infty} f(n)/g(n) = 1$. 
\item Also, $f\approx g$ is to mean that $f$ and $g$ are in the same order if $f(n)/C \leq g(n) \leq Cf(n)$ for some constant $C \geq 1$ independent of $n$. $f \lesssim g$, defined by $\lim_{n\ra\infty} \lp f(n)-g(n) \rp \leq 0$, means that $f(n)$ is asymtotically smaller than $g(n)$. $f \gtrsim g$ is equivalent to $g \lesssim f$. 
These notations are equivalent to the standard Big O notations $\Theta(\cdot)$, $O(\cdot)$, and $\Omega(\cdot)$, which we also use in this paper interchangeably. 
\item $\lnorm\bm{x}\rnorm_2$, $\lnorm\bm{x}\rnorm_1$ is the $\ell_2$ and $\ell_1$ norm for a vector $\bm{x}$, respectively.
\item $d_\mrm{H}(\bm{x},\bm{y})$ is the Hamming distance between two vectors $\bm{x}$ and $\bm{y}$.
\item For a matrix $\mbf{M}$, we denote its operator norm by $\lnorm\mbf{M}\rnorm_{\mrm{op}}$ and its Frobenius norm by $\lnorm\mbf{M}\rnorm_{\mrm{F}}$.
\item For a $d$-dimensional tensor $\mbf{T}$, we denote its $(l_1,l_2,...,l_d)$-th element by $\mrm{T}_{\bm{l}}$, where $\bm{l} \eqDef (l_1,...,l_d)$, and we write $\mbf{T} = \lb \mrm{T}_{\bm{l}}\rb$.
\item Finally, let $\mrm{O}(k_1,k_2) = \lbp \mbf{V}\in \mbb{R}^{k_1 \by k_2} \mid \mbf{V}^{\mrm{T}}\mbf{V} = \mbf{I}_{k_2}\rbp$ for $k_1 \geq k_2$ be the set of all orthogonal $k_1\by k_2$ matrices.
\end{itemize}

\section{Problem Formulation}
\label{sec:formu}
\subsection{Community Relations}
\label{subsec:Kappa_d}

Before introducing the random hypergaph model $d$-{\hSBM}, we first describe the community relations among $d$ nodes, which serves as the basic building block of our model. Let $\mcal{K}_d \eqDef \lbp r_i\rbp$ be the set of all possible community relations under $d$-{\hSBM} and $\kappa_d \eqDef \lba\mcal{K}_d\rba$ denotes the total number of them. In contrast to the dichotomy situation (same community or not) concerning the appearance of an edge between two nodes in the usual symmetric {\SBM}, there is a growing number of community relations in $d$-{\hSBM} as the order $d$ increases. In order not to mess up with them, we use the idea of \emph{majorization} \cite{MarshallOlkin_11} to organize $\mcal{K}_d$ with each $r_i$ in the form of a \emph{histogram}. Specifically, the histogram operator $\msf{hist}(\cdot)$ is used to transform a vector $\bm{r}\in[k]^d$ into its histogram vector $\msf{hist}(\bm{r})$. For convinience, we sort the histogram vector in descending order and append zero's if necessary to make $\msf{hist}(\bm{r})$ a length-$d$ vector. The notion of majorization is introduced as follows. For any $\bm{a},\bm{b}\in\mbb{R}^d$, we say that $\bm{a}$ majorizes $\bm{b}$, written as $\bm{a}\succ\bm{b}$, if $\sum_{i=1}^k a_i^{\da} \geq \sum_{i=1}^k b_i^{\da}$ for $k=1,\ldots,d$ and $\sum_{i=1}^d a_i = \sum_{i=1}^d b_i$, where $x_i^{\da}$'s are elements of $\bm{x}$ sorted in descending order. Observe that each community relation $r_i$ in $\mcal{K}_d$ can be uniquely represented, when sorted in descending order, by a $d$-dimensional histogram vector $h_i$. We arrange the elements in $\mcal{K}_d$ in majorization (pre)order such that $\bm{h}_i\succ\bm{h}_j$ if and only if $i<j$. For example, $r_1$ is relation \emph{all-same} with the most concentrated histogram $\bm{h}_1 = (d,0,\ldots,0)$ and $r_2$ is the \emph{only-1-different} relation with $\bm{h}_2 = (d-1,1,0,\ldots,0)$. Likewise, $r_{\kappa_d-1}$ is the \emph{only-2-same} relation with $\bm{h}_{\kappa_d-1} = (2,1,\ldots,1,0)$ and the last one in $\mcal{K}_d$, relation \emph{all-different} $r_{\kappa_d}$, has a histogram vector $\bm{h}_{\kappa_d} = (1,\ldots,1)$ being the all-one vector.

\begin{example}[$\mcal{K}_4$ in $4$-{\hSBM}]
$\lba\mcal{K}_4\rba=\kappa_4=5$ with histogram vectors being
\begin{center}
\begin{tabular}{clcc}
\multicolumn{2}{c}{\bf \textrm{Relation}} & {\bf \textrm{Histogram}}     & {\bf \textrm{Connecting Probability}} \\
\hline
$r_1$ & $($\emph{all-same}$)$         & $\bm{h}_1=(4,0,0,0)$ & $p_1$ \\
$r_2$ & $($\emph{only-1-different}$)$ & $\bm{h}_2=(3,1,0,0)$ & $p_2$ \\
$r_3$ &                               & $\bm{h}_3=(2,2,0,0)$ & $p_3$ \\
$r_4$ & $($\emph{only-2-same}$)$      & $\bm{h}_4=(2,1,1,0)$ & $p_4$ \\
$r_5$ & $($\emph{all-different}$)$    & $\bm{h}_5=(1,1,1,1)$ & $p_5$ \\
\end{tabular}
\end{center}
\end{example}


\subsection{Random Hypergraph Model: $d$-{\hSBM}}
\label{subsec:d-hSBM}

In a $d$-uniform hypergraph, the adjacency relation among the $n$ nodes in $\mcal{G}=(\mcal{V},\mcal{E})$ can be equivalently represented by a $d$-dimensional $n\by\cdots\by n$ random tensor $\mbf{A} \eqDef \lb\mrm{A}_{\bm{l}}\rb$ (the size of each dimension being $n$), where $\bm{l} = (l_1,\ldots,l_d) \in [n]^d$ is the access index of an element in the tensor. The following two natural conditions on this adjacency tensor come from the basic properties of an undirected hypergraph:
\begin{equation*}
\begin{array}{ll}
\text{No self-loop:} &
\mrm{A}_{\bm{l}} \neq 0 \iff \lba \{l_1,\ldots,l_d\} \rba = d\quad \text{(all $d$ elements in $\bm{l}$ are distinct)}\\
\text{Symmetry:} & 
\mrm{A}_{\bm{l}} = \mrm{A}_{\bm{l}_{\pi}} \;\forall \pi\in\mcal{S}_d.
\end{array}
\end{equation*}
For each $\bm{l}$, $\mrm{A}_{\bm{l}}$ is a Bernoulli random variable with success probability $\mrm{Q}_{\bm{l}}$. The parameter tensor $\mbf{Q} \eqDef \lb\mrm{Q}_{\bm{l}}\rb$ depends only on the community assignments of the associated nodes in the hyperedge and forms a block structure.
The block structure is characterized by a symmetric community connection $d$-dimensional tensor $\mbf{B}\in[0,1]^{k\by\cdots\by k}$ where $\mrm{Q}_{\bm{l}} = \mrm{B}_{\sigma(\bm{l})}$.

To setup the parameter space considered in our statistical study, below we first introduce some further notations.
Let $n_t = \lba\{ i\mid\sigma(i)=t \}\rba$ be the size of the $t$-th community for $t\in[k]$. Besides, let $\bm{p} = (p_1,\ldots,p_{\kappa_d}) \in(0,1)^{\kappa_d}$ where $p_i$ denotes the success probability of the Bernoulli random variable that corresponds to the appearance of a hyperedge with relation $r_i\in\mcal{K}_d$.
We make a natural assumption
that $p_i\geq p_j \;\forall i<j$. The more concentrated a group is, the higher the chances that the members will be connected by an hyperedge.

\begin{remark}
We would like to note that there is nothing peculiar about the assumption that $p_i$'s are in decreasing order and the condition can be relaxed. All that is required are that the connecting probabilities $p_i$'s are well separated and the difference between each $p_i-p_j$ are within the same order. See \autoref{sec:theo} for a more formal statement of our main result.
\end{remark}

The parameter space considered here is a \emph{homogeneous} and \emph{approximately equal-sized} case where each $n_t \approx \lfl\frac{n}{k}\rfl$. Formally speaking (let $n^{\prime} \eqDef \lfl\frac{n}{k}\rfl$),
\begin{equation}
\label{eq:theta_d_0}
\Theta_d^0(n,k,\bm{p},\eta) \eqDef \Big\{ (\mbf{B},\sigma) \mid \sigma: [n]\ra[k], \; n_t \in \lb(1-\eta)n^{\prime},(1+\eta)n^{\prime}\rb \;\forall t \in[k] \Big\}
\end{equation}
where $\mbf{B}$ has the property that $\mrm{B}_{\sigma(\bm{l})} = p_i$ if and only if $\msf{hist}(\sigma(\bm{l})) = \bm{h}_i$. In other words, only the histogram of the community labels within a group matters when it comes to connectivity. $\eta$ is a parameter that controls how much $n_t$ could vary. We assume the more interesting case that $\eta\geq\frac{1}{n^{\prime}}$ where the community sizes are not restricted to be exactly equal. Interchangeably, we would write $\bm{l}\overset{\sigma}{\sim}r_i$ to indicate the community relation within nodes $l_1,\ldots,l_d$ under the assignment $\sigma$. Throughout the paper, we will assume that the order $d$ of the observed hypergraph is a constant, while the other parameters, including the total number of communities $k$ and the hyperedge connection probability $\bm{p}$, can be coupled with $n$. Specifically, $k$ can either be a constant or it can also scale with $n$. Moreover, as pointed out in \cite{ISIT17_CDHG}, the regime where the hypergraph {\SBM} is weakly recoverable could be orderly lower than the one considered in {\SBM} of graphs \cite{AbbeSandon_15}. To guarantee the solvability of weak recovery in $d$-{\hSBM}, we set the probability parameter $\bm{p}$ should at least be in the order of $\Omega\lp\frac{1}{n^{d-1}}\rp$. Therefore, we would write $\bm{p} = \frac{1}{n^{d-1}} (a_1,\ldots,a_{\kappa_d})$ where $a_i=\Omega(1)$ for all $i=1,\ldots,\kappa_d$. We would like to note that the probability regime $\Omega\lp\frac{1}{n^{d-1}}\rp$ considered here is first motivated in \cite{ISIT17_CDHG}. Under $3$-{\hSBM}, the authors in \cite{ISIT17_CDHG} consider $\bm{p} = \Omega\lp\frac{1}{n^2}\rp$ for the probability parameter, which is orderly lower than the one ($\Theta\lp\frac{1}{n}\rp$) required for partial recovery in \cite{AbbeSandon_15} and the minimax risk in \cite{ZhangZhou_16} for graph {\SBM}. The motivation is that, since the total number of random variables in a random $3$-uniform hypergraph is roughly $n$-times larger than those in a traditional random graph, the underlying hypergraph is allowed to be $n$-times sparser and still retain a risk of the same order. In light of this, we relax the probability parameter $\bm{p}$ from $\Omega\lp\frac{1}{n}\rp$ to $\Omega\lp\frac{1}{n^{d-1}}\rp$ in $d$-{\hSBM}.

\subsection{Performance Measure}
\label{subsec:def_Rstar}

To evaluate how good an estimator $\what{\sigma}:\mcal{G}\ra[k]^n$ is, we use the \emph{mismatch ratio} as the performance measure to the community detection problem.
The un-permuted loss function is defined as
\begin{equation*}
\loss_0(\sigma_1,\sigma_2) \eqDef \frac{1}{n}\> d_{\mrm{H}}(\sigma_1,\sigma_2)
\end{equation*}
where $d_{\mrm{H}}$ is the Hamming distance. It directly counts the proportion of misclassified nodes between an estimator and the ground truth assignment. Concerning the issue of possible re-labeling, the mismatch ratio is defined as the loss function which maximizes the agreements between an estimator and the ground truth after an alignment by label permutation.
\begin{equation}
\label{eq:mm_ratio}
\loss(\what{\sigma},\sigma) \eqDef \min_{\pi\in\mcal{S}_k} \loss_0(\what{\sigma}_{\pi},\sigma)
\end{equation}
As convention, we use $\Risk_{\sigma}(\what{\sigma}) \eqDef \E_{\sigma}\loss(\what{\sigma},\sigma)$ to denote the corresponding risk function.
Finally, the minimax risk for the parameter space $\Theta_d^0(n,k,\bm{p},\eta)$ under $d$-{\hSBM} is denoted as
\begin{equation*}
\Risk_d^{\ast} \eqDef \adjustlimits\inf_{\what{\sigma}}\sup_{(\mbf{B},\sigma)\in\Theta_d^0} \Risk_{\sigma}(\what{\sigma}).
\end{equation*}

\begin{remark}
Notice that in a symmetric (homogeneous) {\SBM} \cite{Abbe_17}, the connectivity tensor $\mbf{B}$ is uniquely determined by the labeling function $\sigma$. Therefore, 
we would drop the subscript $\mbf{B}$ in $\Pr_{\mbf{B},\sigma}\lbp\cdot\rbp$ and write $\Pr_{\sigma}\lbp\cdot\rbp$ when it comes to the uncertainty arising from the random hypergraph model with the underlying assignment being $\sigma$. Similarly, we would write $\E_{\sigma}\lb\cdot\rb$ instead of $\E_{\mbf{B},\sigma}\lb\cdot\rb$ for ease of notation.
\end{remark}

\section{Prior Works}
\label{sec:moti}
For the case $d=2$, the asymptotic minimax risk $\Risk_2^{\ast}$ is characterized in \cite{ZhangZhou_16}, which decays to zero  exponentially fast as $n\ra\infty$. In addition, the (negative) exponent of $\Risk_2^{\ast}$ is determined by the R\'{e}nyi divergence of order $1/2$ between two Bernoulli distributions $\Ber(p)$ and $\Ber(q)$
\begin{equation}
\label{eq:Renyi_div_Ber}
I_{pq} \eqDef -2\log \lp \sqrt{pq}+\sqrt{1-p}\sqrt{1-q} \rp
\end{equation}
where $p$ is the success probability of a same-community edge while $q$ stands for a different-community one. Extending from traditional graph {\SBM} to a hypergraph setting, the authors in \cite{ISIT17_CDHG} generalize the minimax result obtained in \cite{ZhangZhou_16} to the $3$-{\hSBM} model as follows
\begin{equation*}
-\log \Risk_3^{\ast} \asymp \binom{n^{\prime}}{2} I_{pq} + (k-2)(n^{\prime})^2 I_{qr}
\end{equation*}
where the probability parameter $\bm{p}=(p,q,r)$ corresponds to the community relations with histograms $(3,0,0)$, $(2,1,0)$ and $(1,1,1)$, respectively. Observe that the exponent of the minimax risk in $3$-{\hSBM} does not depend on the divergence term $I_{pr}$ explicitly. That is, $\Risk_3^{\ast}$ consists of only those neighboring divergence terms whose histogram vectors have a $\ell_1$ distance of $2$. Besides, associated with each divergence term $I_{p_i p_j}$ is a weighted coefficient, i.e. $\binom{n^{\prime}}{2}$ for $I_{pq}$ and $(k-2)(n^{\prime})^2$ for $I_{qr}$. These coefficients appears in the hypothesis testing problem when deriving the lower bound of the minimax result. Essentially, they represent the total number of random variables that appear either as a relation-$r_i$ hyperedge or as a relation-$r_j$ hyperedge when the community label of this targeted node is being tested.

It turns out that the optimal minimax risk in $d$-{\hSBM} also decays to zero exponentially fast, given that the outcome of the initialization algorithm satisfies a certain condition. The exponent, as stated formally later, is a weighted combination of divergence terms. To specify the weight in this weighted average, we introduce further notations below. We use $\mcal{N}_d \eqDef \lbp(r_i,r_j) \mid i<j, \; \lnorm \bm{h}_i-\bm{h}_j\rnorm_1 = 2\rbp$ to denote the collection of ordered pairs of relations in $\mcal{K}_d$ that are \emph{neighbors} to each other. Second, there is a combinatorial number associated with every pairwise divergence term.
Precisely, let us consider a least favorable sub-parameter space of $\Theta_d^0$.
\begin{equation}
\label{eq:ThetaL}
\Theta_d^L(n,k,\bm{p},\eta) \eqDef \Big\{ (\mbf{B},\sigma)\in\Theta_d^0 \mid n_t \in \lbp n^{\prime}-1,n^{\prime},n^{\prime}+1 \rbp \;\forall t\in[k], \ n_{\sigma(1)} = n^{\prime}+1 \Big\}
\end{equation}
In $\Theta_d^L$, each community takes on only three possible sizes. In addition, there are exactly $n^{\prime}+1$ members in the community where the first node belongs. We pick a $\sigma_0$ in $\Theta_d^L$ and construct a new assignment $\sigma[\sigma_0]$ based on $\sigma_0$:
\begin{equation*}
\sigma[\sigma_0](i) =
\begin{cases}
\argmin_{2\leq t\leq k} \lbp n_t = n^{\prime} \rbp & \text{, for } i=1 \\
\sigma_0(i)                                        & \text{, for } 2\leq i\leq n.
\end{cases}
\end{equation*}
In other words, assignments $\sigma[\sigma_0]$ and $\sigma_0$ only disagree on the label of the first node. For each pair $(r_i,r_j)$ in $\mcal{N}_d$, we define the weighted coefficient
\begin{equation*}
m_{r_i r_j} \eqDef \Big|\big\{ \bm{l}=(1,l_2,\ldots,l_d) \mid \bm{l}\overset{\sigma_0}{\sim}r_i, \bm{l}\overset{\sigma[\sigma_0]}{\sim}r_j \big\}\Big|
\end{equation*}
as the number of relation-$r_i$ hyperedges that we mistake as relation-$r_j$ hyperedges. Note that the above definition is independent to the choice of $\sigma_0\in\Theta_d^L$ due to the community size constraints.

\begin{example}[$\mcal{N}_4$ in $4$-{\hSBM}]
\label{ex:mrirj}
$\lba\mcal{N}_4\rba=5$ with elements
\begin{center}
\begin{tabular}{cc}
{\bf \textrm{Relation Pair}} & {\bf \textrm{Weighted Coefficient}} \\
\hline
$(r_1,r_2)$ & $m_{r_1 r_2} \asymp \binom{n^{\prime}}{3}$ \\
$(r_2,r_3)$ & $m_{r_2 r_3} \asymp \binom{n^{\prime}}{2}n^{\prime}$ \\
$(r_3,r_4)$ & $m_{r_3 r_4} \asymp n^{\prime}(k-2)\binom{n^{\prime}}{2}$ \\
$(r_2,r_4)$ & $m_{r_2 r_4} \asymp \binom{n^{\prime}}{2}(k-2)n^{\prime}$ \\
$(r_4,r_5)$ & $m_{r_4 r_5} \asymp n^{\prime}\binom{k-2}{2}(n^{\prime})^2$ \\
\end{tabular}
\end{center}
Note that $m_{r_1 r_2}$ is the smallest while $m_{r_4 r_5}$ is the largest.
\end{example}


\section{Main Contribution}
\label{sec:main}
The optimal minimax risk for the homogeneous and approximately equal-sized parameter space $\Theta_d^0(n,k,\bm{p},\eta)$ under the probabilistic model $d$-{\hSBM} is characterized as follows.

\begin{mythm}{Main Theorem}
\label{thm:minimax_risk}
Suppose $\bm{p}=\frac{1}{n^{d-1}}(a_1,\cdots,a_{\kappa_d}) = \Omega\lp\frac{1}{n^{d-1}}\rp$ and $a_i\approx a_j \;\forall i,j=1,\ldots,\kappa_d $. If
\begin{equation}
\label{eq:main_condi_main}
\frac{\sum_{i<j:(r_i,r_j)\in\mcal{N}_d} m_{r_i r_j}I_{p_i p_j}}{k^d \log k} \ra \infty
\end{equation}
and
\begin{equation}
\label{eq:main_condi_order}
\frac{\sum_{i<j:(r_i,r_j)\in\mcal{N}_d} m_{r_i r_j}I_{p_i p_j}}{k^d \max_{(i,j):(r_i,r_j)\in\mcal{N}_d} \frac{p_1-p_{\kappa_d}}{p_i-p_j}} \ra \infty
\end{equation}
as $n\ra\infty$. Then
\begin{equation}
\label{eq:minimax_result}
\log \Risk_d^{\ast} \asymp - \sum_{i<j:(r_i,r_j)\in\mcal{N}_d} m_{r_i r_j}I_{p_i p_j}.
\end{equation}
If $k$ is a constant, then \eqref{eq:minimax_result} holds without the further assumption \eqref{eq:main_condi_order}.
\end{mythm}

\begin{remark}
In this work, we assume that the order of the hypergraph $d$ is a constant. More generally, one may also wonder how the characterized minimax risk changes when this order $d$ is also allowed to scale with $n$. Certainly, the expression above for the optimal minimax risk depends on the hypergraph order $d$, yet only implicitly. To obtain an explicit form of $\Risk_d^{\ast}$ in terms of $d$, we have to get further estimates of those weighted coefficient $m_{r_i r_j}$'s as well as the corresponding R\'{e}nyi divergence term $I_{p_i p_j}$'s. The latter can be estimated by $I_{p_i p_j}\approx\frac{(a_i-a_j)^2}{n^{d-1}a_i}$ when $a_i\approx a_j=\omega\lp\frac{1}{n^{d-1}}\rp$ as assumed in the main theorem. On the other hand, as commented in Example~\ref{ex:mrirj}, it is not hard to see that $m_{r_i r_j}$ achieves its minimum at $\binom{n^{\prime}}{d-1}$ between $r_1$ with $\bm{h}_1=(d,0,\ldots,0)$ and $r_2$ with $\bm{h}_2=(d-1,1,0,\ldots,0)$ while it attains its maximum at $n^{\prime}\binom{k-2}{d-2}\lp n^{\prime}\rp^{d-2}$ between $r_{\kappa_d-1}$ with $\bm{h}_{\kappa_d-1}=(2,1,\ldots,1,0)$ and $r_{\kappa_d}$ with $\bm{h}_{\kappa_d}=(1,\ldots,1)$. Therefore, when the differences $a_i-a_j$ are constant $\forall i<j:(r_i,r_j)\in\mcal{N}_d$, the last term with $m_{r_{\kappa_d-1} r_{\kappa_d}}$ dominates other terms in the summation $\sum_{i<j:(r_i,r_j)\in\mcal{N}_d}m_{r_i r_j}I_{p_i p_j}$ seeing that the parameter $k$ is coupled with $n$. In particular, the error exponent for the optimal minimax risk in equation \eqref{eq:minimax_result} is in the order of $\frac{1}{(d-2)!}$. 
Surprisingly, the minimax risk would decay more slowly due to the factorial term in the denominator as the order $d$ increases. However, we would like to note that this observation is valid only under the assumption that the considered hypergraphs are in the \emph{sparse} regime $\Theta\lp\frac{1}{n^{d-1}}\rp$ where there are roughly $n\log n$ hyperedges generated no matter how large the order $d$ is. 
\end{remark}

The minimax risk is provably achieved, through Theorem~\ref{thm:theo_upper} in Section~\ref{sec:theo}, by the proposed two-step algorithm. Roughly speaking, we first demonstrate that the second-step refinement is capable of obataining an acurate parameter estimation as long as the first-step initialization satisfies a weak consistency condition. Then, the local MLE step is proved to achieve a mismatch ratio as the desired minimax risk, with which the local majority voting could recover the true community label for each node with the guaranteed risk. Finally, we show that our proposed spectral clustering algorithm with the hypergraph Laplacian matrix are qualified as a first-step initialization algorithm. We will compare our theoretical findings to those for the graph case \cite{GaoMa_15} below as well as for the hypergraph setting \cite{GhoshdastidarDukkipati_17} later in Section~\ref{sec:theo}. On the other hand, the converse part is established through Theorem~\ref{thm:lower} in Section~\ref{sec:conv}.

\subsection{Implications to Exact Recovery}
\label{subsec:main_imp}

Since we consider a minimax framework, the theoretical guarantees of our two-step algorithm are also sufficient to ensure the partial recovery and the exact recovery as considered under the Bayesian perspective \cite{AbbeSandon_15}. Before presenting the theorems in regard to the community recovery in the Bayesian case, let's first refresh on these two recovery notions. The definitions of different recovery criterions discussed here can be found in the comprehensive survey \cite{Abbe_17}. We paraphrase them below for completeness. Please refer to the survey for more details and the references therein. In terms of the mismatch ratio \eqref{eq:mm_ratio},

\begin{definition}[Revised Definition~4 in \cite{Abbe_17}]
Consider a $\sigma\in\Theta_d(n,k,\bm{p},\eta)$ and a corresponding random hypergraph $G$. The following recovery requirements are solved if there exists an algorithm $A$ which takes $G$ as input and estimates $\what{\sigma}=A(G)$ such taht
\begin{itemize}
\item \textbf{Partial Recovery:} $\Pr\lbp \ell(\what{\sigma},\sigma)\leq\alpha\rbp=o(1),\,\alpha\in(0,1)$
\item \textbf{Exact Recovery:} $\Pr\lbp \ell(\what{\sigma},\sigma)=0\rbp=o(1)$
\end{itemize}
where the probability is taken over the random realizations of $G$ and the asymptotic notation is with respect to the growth of $n$.
\end{definition}

Our proposed two-step algorithm in Section~\ref{sec:algo} can provably satisfy the exact recovery criterion.
\begin{theorem}
\label{thm:exact_reco}
If
\begin{equation}
\label{eq:exact_reco_condi}
\liminf_{n\ra\infty}\frac{\sum_{i<j:(r_i,r_j)\in\mcal{N}_d} m_{r_i r_j}I_{p_i p_j}}{\log n} > 1,
\end{equation}
then Algorithm~\ref{alg:refine} combined with Algorithm~\ref{alg:spec_init} is able to solve the exact recovery problem.
\end{theorem}

\begin{proof}
With \eqref{eq:exact_reco_condi}, for any $n\in\mbb{N}$ there exists a small constant $c>0$ such that
\begin{equation*}
\frac{\sum_{i<j:(r_i,r_j)\in\mcal{N}_d} m_{r_i r_j}I_{p_i p_j}}{\log n} > 1 + c.
\end{equation*}
By the Markov inequality, we have
\begin{align*}
\Pr\lbp \ell(\what{\sigma},\sigma)<\frac{1}{n}\rbp &\leq n\cdot\E\lb \ell(\what{\sigma},\sigma)\rb \\
 &\leq n\Risk_d^{\ast} < n^{-c}.
\end{align*}
Note that the event that mismatch ratio is smaller than $\frac{1}{n}$ is equivalent to the event that it is identical to $0$.
\end{proof}

We would like to note that partial recovery is immediate from the exact recovery. Indeed, the required condition can be relaxed from \eqref{eq:exact_reco_condi} and depends on the extent of distortion $\alpha$.

\subsection{Comparison with \cite{GaoMa_15}}
\label{subsec:main_comp}

We can recover the minimax result obtained in \cite{GaoMa_15} by specializing $d=2$ in the main theorem above. In \cite{GaoMa_15}, the authors consider the traditional {\SBM} model under a homogeneous parameter space with connecting probability being $\bm{p}=(p,q)=(\frac{a}{n},\frac{b}{n})$. We would like to note that the parameter space considered in \cite{GaoMa_15} is more general in the sense that it needs not be nearly equal-sized case. To be more specific, the size of each community $n_t$ is allowed to vary within $\lb \frac{1}{\beta}n^{\prime}, \beta n^{\prime} \rb$ (where $n^{\prime}=\lfloor\frac{n}{k}\rfloor$). However, the parameter $\beta$ controlling this variation is itself only restricted in the range $[1,\sqrt{5/3}]$ for some technical issue, which makes the attempted relaxation on the community size less interesting. In light of this, we compare only the minimax result in \cite{GaoMa_15} with $\beta=1$, which is $\Theta_2^0$ in our notation. The overall result for $\Theta_2^0(n,k,\bm{p},\eta)$, when combining the spectral initialization step with the local refinement step proposed therein (denoted as $\what{\sigma}_{\mrm{GM}}$), can be summarized as follows: Suppose $\bm{p}=(\frac{a}{n},\frac{b}{n})$ and $a\approx b$. If
\begin{equation}
\label{eq:Gao_condi}
\frac{(a-b)^2}{ak^3\log k} \ra \infty
\end{equation}
as $n\ra\infty$. Then, there exists a sequence $\zeta_n \ra0$ such that
\begin{equation}
\label{eq:Gao_result}
\sup_{(\mbf{B},\sigma)\in\Theta_2^0} \Pr_{\sigma}\Big\{ \loss(\what{\sigma}_{\mrm{GM}},\sigma) \geq \exp\big( -(1-\zeta_n)n^{\prime}I_{pq} \big)\Big\} \ra 0.
\end{equation}

Indeed, conditiion \eqref{eq:Gao_condi} required is exactly the same as \eqref{eq:main_condi_main} by using the approximation $I_{pq}\approx\frac{(a-b)^2}{na}$. Note that there is only one community relation pair $(r_1,r_2)$ in $\mcal{N}_2$ and the weighted coefficient is $m_{r_1 r_2} = n^{\prime}\asymp\frac{n}{k}$. In fact, the situation is very simple in {\SBM} since there are only two possible community relations, i.e. intra-community (relation \emph{all-same}) and inter-community (relation \emph{all-diff}). On the contrary, the relational information gets more and more complicated as $d$ increases. This inevitable ``curse of dimension" is reflected in the second assumption \eqref{eq:main_condi_order} we made in the main theorem.
First, recall that we set all of the probability parameters $\bm{p}=(p_1,\ldots,p_{\kappa_d})$ in the same order as the condition $a\approx b$ required in \cite{GaoMa_15}. Apart from that, we also need to make sure that the differences $p_i-p_j$ associated with the pairs $(r_i,r_j)\in\mcal{N}_d$ remain in the same order to successfully upper-bound the error probability in the proof of achievablity. Under the traditional {\SBM}, it is not hard to see that the assumption \eqref{eq:main_condi_order} is weaker than the assumption \eqref{eq:main_condi_main}. Therefore, the overall requirement is equivalent to \eqref{eq:Gao_condi} made in \cite{GaoMa_15} without any further assumption.

\section{Proposed Algorithms}
\label{sec:algo}

In this section, we propose our main algorithm for community detection in random hypergraphs, which is later in \autoref{sec:theo}\
proved to achieve the minimax risk $\Risk_d^{\ast}$ in the $d$-{\hSBM} asymptotically.
The algorithm (Algorithm~\ref{alg:refine}) comprises two major steps. In the first step, for each $u\in[n]$, it generate an estimated assignment $\what{\sigma}_u$ of all nodes except $u$ by applying an initialization algorithm $\msf{Alg}_{init}$ on the sub-hypergraph without the vertex $u$. For example, we can apply the hypergraph clustering method described in \autoref{subsec:algo_step1} on $\mbf{A}_{-u}$, the $(n-1)\by\cdots\by(n-1)$ sub-tensor of $\mbf{A}$ when the $u$-th coordinate is removed in each dimension. Then, in the second step, the label of $u$ under $\what{\sigma}_u$ is determined by maximizing a local likelihood function described in \autoref{subsec:algo_step2}. Note that the parameters of the underlying $d$-{\hSBM} need not be known in advance, as it could conduct a parameter estimation before computing the local likelihood function if necessary. Finally, with $n$ estimated assignments $\{\what{\sigma}_u : u\in[n]\}$, the algorithm combines all of them together and forms a consensus via majority neighbor voting.

\subsection{Refinement Scheme}
\label{subsec:algo_step2}

Let us begin with the \emph{global} likelihood function defined as follows. Let
\begin{equation}
\label{eq:loglike_global}
L(\sigma;\mbf{A}) \eqDef \adjustlimits\sum_{\{i\mid r_i\in\mcal{K}_d\}}\sum_{\{\bm{l}\mid\bm{l}\overset{\sigma}{\sim}r_i\}} \Big( \log p_i\mrm{A}_{\bm{l}} + \log(1-p_i)(1-\mrm{A}_{\bm{l}}) \Big)
\end{equation}
denote the log-likelihood of an adjacency tensor $\mbf{A}$ when the hidden community structure is determined by $\sigma$. For each $u\in[n]$, we use
\begin{equation}
\label{eq:loglike_local}
L_u(\sigma,t;\mbf{A}) \eqDef \adjustlimits\sum_{\{i\mid r_i\in\mcal{K}_d\}}\sum_{\{\bm{l}\mid l_1=u, \bm{l}\overset{\sigma}{\sim}r_i\}} \Big( \log p_i\mrm{A}_{\bm{l}} + \log(1-p_i)(1-\mrm{A}_{\bm{l}}) \Big)
\end{equation}
to denote those likelihood terms in \eqref{eq:loglike_global} pertaining to the $u$-th node when its label is $t$.
It is not hard to see that $L_u(\sigma,t;\mbf{A})$ is a sum of independent Bernoulli random variables. However, $L_u(\sigma,t;\mbf{A})$ is not independent of $L_v(\sigma,s;\mbf{A})$ for any $u\neq v$ since those random hyperedges that might enclose vertex $u$ and vertex $v$ simultaneously appear in both of the summands of the likelihood terms. The global likelihood function and the local likelihood function is related by
\begin{equation*}
L(\sigma;\mbf{A}) = \frac{1}{d}\sum_{u\in[n]}L_u(\sigma,t;\mbf{A}).
\end{equation*}
This is because each likelihood term in \eqref{eq:loglike_global} is counted exactly $d$ times when summing over all possible equation \eqref{eq:loglike_local}'s.
For each node $u\in[n]$, based on the estimated assignment of the other $n-1$ nodes, we use the following \emph{local MLE method} to predict the label of $u$.
\begin{equation*}
\what{\sigma}_u(u) \eqDef \argmax_{t\in[k]} L_u(\sigma,t;\mbf{A})
\end{equation*}
When the connectivity tensor $\mbf{B}$ that governs the underlying random hypergraph model $d$-{\hSBM}
is unknown when evaluating the likelihood, we will use $\what{L}(\sigma;\mbf{A})$ and $\what{L}_u(\sigma,t;\mbf{A})$ to denote the global and local likelihood function with the true $\mbf{B}$ replaced by its estimated counterpart $\what{\mbf{B}}$. Since the presence of each edge is independent based on our probabilistic model, we use the sample mean $\what{\mbf{B}}^u$ to estimate the real parameters. Note that the superscript $u$ is to indicate the fact that the estimation $\what{\mbf{B}}^u$ is calculated with node $u$ taken out.
Finally, consensus is drawn by using the majority neighbor voting.
In fact, the consensus step looks for a consensus assignment for the $n$ possible different community assignments obtained in the local MLE method in Algorithm~\ref{alg:refine}. Since all these $n$ assignments will be close to the ground truth up to some permutation, this step combines all of them to conclude a single community assignment as the final output.


\begin{algorithm2e}
\label{alg:refine}
\caption{Main Algorithm}
\SetAlgoLined
\DontPrintSemicolon
\SetKwInOut{KwInAligned}{Input}
\SetKwInput{KwPNV}{Local MLE}
\SetKwInput{KwCSS}{Consensus}
    \KwInAligned{Observation tensor $\mbf{A}\in\{0,1\}^{n\by\cdots\by n}$,\\
        number of communities $k$,\\
        initialization algorithm $\msf{Alg}_{init}$.
    }
    \KwPNV{\;
        \For{$u=1$ \KwTo $n$}{
            Apply $\msf{Alg}_{init}$ on $\mbf{A}_{-u}$ to obtain $\what{\sigma}_u(v)\;\forall v\neq u$.\;
            Define $\wtild{C}_i^u=\{v\vert \what{\sigma}_u(v)=i\}$ for all $i\in[k]$.\;
            Estimate entries of $\mbf{B}$ by its sample mean $\what{\mbf{B}}^u$. Specifically, estimate each probability parameter $p_i$ with\;
            \begin{equation*}
            \what{p}_i = \frac{\sum_{\bm{l}}\mrm{A}_{\bm{l}}\Indc{\bm{l}\overset{\what{\sigma}_u}{\sim}r_i}}{\sum_{\bm{l}}\Indc{\bm{l}\overset{\what{\sigma}_u}{\sim}r_i}}\;
            \end{equation*}
            Assign the label of node $u$ according to
            \begin{equation}
            \label{eq:local_MLE}
            \what{\sigma}_u(u) = \argmax_{t\in[k]} \what{L}_u(\what{\sigma}_u,t;\mbf{A})\;
            \end{equation}
        }
    }
    \KwCSS{Define $\what{\sigma}(1)=\what{\sigma}_1(1)$. For $u=2,\ldots,n$, define
        \begin{equation}
        \label{eq:css_method}
        \what{\sigma}(u) = \argmax_{t\in[k]}\big| \{v\vert\what{\sigma}_1(v) = t\} \cap \{v\vert\what{\sigma}_u(v) = \what{\sigma}_u(u)\} \big|\;
        \end{equation}
    }
    \KwOut{Community assignment $\what{\sigma}_2$.}
\end{algorithm2e}

\subsection{Spectral Initialization}
\label{subsec:algo_step1}

In order to devise a good initialization algorithm $\msf{Alg}_{init}$, we develop a hypergraph version of the unnormalized spectral clustering \cite{VonLuxburg_07} with regularization \cite{ChinRao_15}.
In particular, a modified version of the \emph{hypergraph Laplacian} described below is employed. Let $\mbf{H} = \lb\mrm{H}_{ve}\rb$ be the $\lba\mcal{V}\rba \by \lba\mcal{E}\rba$ incidence matrix, where each entry $\mrm{H}_{ve}$ is the indicator function whether or not node $v$ belongs to the hyperedge $e$. Note that the incidence matrix $\mbf{H}$ contains the same amount of information as the the adjacency tensor $\mbf{A}$. Let
\begin{equation*}
d_u \eqDef \sum_{e\in\mcal{E}} \mrm{H}_{ue}
\end{equation*}
denote the degree of the $u$-th node, and
\begin{equation*}
\bar{d} \eqDef \frac{1}{n} \sum_{u \in [n]} d_{u}
\end{equation*}
be the average degree across the hypergraph. The \emph{unnormalized} hypergraph Laplacian is defined as
\begin{equation}
\label{eq:hyper_Lapla}
\mcal{L}(\mbf{A}) \eqDef \mbf{H}\mbf{H}^{\mrm{T}}-\mbf{D}
\end{equation}
where $\mbf{D} = \msf{diag} (d_{1},\ldots,d_{n})$ is a diagonal matrix representing the degree distribution in the hypergraph with adjacency tensor $\mbf{A}$ and $(\cdot)^{\mrm{T}}$ is the usual matrix transpose. Note that $\mcal{L}(\mbf{A})$ can be thought of as an encoding of the higher-dimensional connectivity relationship $\mbf{A}$ into a two-dimensional matrix.

Before we directly apply the spectral method, high-degree abnormals in the tensor $\mbf{A}$ is first trimmed to ensure the performance of the clustering algorithm. Specifically, we use $\mbf{A}_{\tau}$ to denote the modification of $\mbf{A}$ where all coordinates pertaining to the set $\{u\in[n]\mid d_{u}\geq\tau\}$ are replaced with all-zero vectors. Let $\mbf{H}_{\tau}$ and $\mbf{D}_{\tau}$ be the corresponding incidence matrix and degree matrix of $\mbf{A}_{\tau}$, respectively. The spectrum we are looking for is the trimmed version of $\mcal{L}$, denoted as
\begin{equation}
\label{eq:trimmed_Laplacian}
\msf{T}_{\tau}(\mcal{L}(\mbf{A})) \eqDef \mbf{H}_{\tau}\mbf{H}_{\tau}^{\mrm{T}} - \mbf{D}_{\tau}
\end{equation}
where the operator $\msf{T}_{\tau}(\cdot)$ represents the trimming process with a degree threshold $\tau$. We use
\begin{equation*}
\msf{SVD}_k\lp \msf{T}_{\tau}(\mcal{L}(\mbf{A}))\rp \eqDef \what{\mbf{U}} = \lb\bm{u}_1^{\mrm{T}} \; \cdots \; \bm{u}_n^{\mrm{T}}\rb^{\mrm{T}} \in \mbb{R}^{n\by k}
\end{equation*}
to denote the $k$ leading singular vectors generated from the singular value decomposition of the trimmed matrix $\msf{T}_{\tau}(\mcal{L}(\mbf{A}))$.
Note that in a conventional spetral clustering algorithm, each node $i\in[n]$ is represented by a reduced $k$-dimensional row vector $\bm{u}_i$.
The spectral clustering algorithm is described in Algorithm~\ref{alg:spec_init}.

\begin{algorithm2e}
\label{alg:spec_init}
\caption{Spectral Initialization}
\SetKwInput{KwCLS}{Cleanup}
\SetAlgoLined
\DontPrintSemicolon
\SetKwInOut{KwInAligned}{Input}
    \KwInAligned{Spectrum $\msf{SVD}_k\lp \msf{T}_{\tau}(\mcal{L}(\mbf{A}))\rp = \lb\bm{u}_1^{\mrm{T}} \; \cdots \; \bm{u}_n^{\mrm{T}} \rb^{\mrm{T}}$,\\
        number of communities $k$,\\
        cirtical radius $r=\mu\sqrt{\frac{k}{n}}$ with some $\mu>0$.
    }
    Set $S=[n]$.\;
    \For{$t=1$ \KwTo $k$}{
        Let $s_t = \argmax_{i\in S} \big|\{ j\in S: \lnorm\bm{u}_j-\bm{u}_i\rnorm_2 < r \}\big|$.\;
        Set $\what{C}_t = \{ j\in S: \lnorm\bm{u}_j-\bm{u}_{s_t}\rnorm_2 < r \}$.\;
        Label $\what{\sigma}(i)=t \;\forall i\in\what{C}_t$.\;
        Update $S \la S\setminus\what{C}_t$.\;
    }
    \KwCLS{If $S\neq\varnothing$, then for any $i\in S$, set
        \begin{equation*}
        \label{eq:spec_init_cls}
        \what{\sigma}(i) = \argmin_{t\in[k]} \frac{1}{\big|\what{C}_t\big|} \sum_{j\in\what{C}_t} \lnorm\bm{u}_j-\bm{u}_i\rnorm_2\;
        \end{equation*}
    }
    \KwOut{Community assignment $\what{\sigma}_1$.}
\end{algorithm2e}

Similar to classical spectral clustering, we make use of the row vectors of $\what{\mbf{U}}$ to cluster nodes. In each loop, we first choose the node which covers the most nodes with radius $r$ in $S$ to be the clustering center. Then, we assign all nodes whose distance from this center is smaller than $r$ to this cluster. At the end of the loop, we remove all nodes within this cluster from $S$. The final cleanup step \eqref{eq:spec_init_cls} in the algorithm is to assign those nodes that deviate too much from all $k$ clusters. It assigns each remaining node to the cluster between which it has the minimum average distance.

\begin{remark}
It is noteworthy that Algorithm~\ref{alg:spec_init} is just one method which is eligible to serve as a qualified first-step estimator $\msf{Alg}_{init}$. As mentioned above, the minimax risk is asymptotically achievable with Algorithm~\ref{alg:refine} as long as the initialization algorithm does not mis-classify too many nodes. The weak consistency requirement is stated explicitly in \autoref{sec:theo} when theoretical guarantees are discussed.
\end{remark}

\subsection{Time Complexity}
\label{subsec:algo_time}

Algorithm~\ref{alg:spec_init} has a time complexity of $O(n^3)$, the bottleneck of which being the $\msf{SVD}_k$ step. Still, the computation of $\msf{SVD}$ could be done approximately in $O(n^2\log n)$ time with high probability \cite{YunProutiere_16} if we are only interested in the first $k$ spectrums. As for the refinement scheme, the sparsity of the underlying hypergraph can be utilized to reduce the complextiy since the whole network structure could be stored in the incidence matrix $\mbf{H}$ equivalently as in the $d$-dimensional adjacency tensor $\mbf{A}$. As a result, the parameter estimation stage only requires $O(dm)$ where $m=|\mcal{E}|$ is the total number of hyperedges realized. Similary, the time complexity would be $O(kdm)$ and $O(kn^2)$ for the calculation of likelihood function and the consensus step, respectively. Hence, the overall complexity for Algorithm~\ref{alg:refine} and Algorithm~\ref{alg:spec_init} combined are $O(n^3\log n+nkm+kn^2)$ for a constant order $d$. It further reduces to $O(n^3\log n)$ in the sparse regime $\bm{p}=O(\log n/n^{d-1})$ where $m=O(n\log n)$ with high probability. 

\begin{remark}
It is possible to simplify our algorithm in the same way as in \cite{GaoMa_15}, where the SVD is done only once. The time complexity of the simplified version of our algorithm will be $O(n^2\log n)$ in the sparse regime. This is comparable to the any other state-of-art min-cut algorithm, which usually exhibit time complexity at least $O(|V||E|)$. 
Although we are not able to provide any theoretical guarantee for this simplified version, as in \cite{GaoMa_15}, empirically it seems to have the same performance as the original algorithm. Proving its asymptotic optimality is left as future work. 
\end{remark}

\section{Theoretical Guarantees}
\label{sec:theo}


Combining the first-step and second-step algorithm in Section~\ref{sec:algo}, we have the following overall performance guarantee which serves as the achievability part of the Main Theorem in Section~\ref{sec:main}.

\begin{theorem}
\label{thm:theo_upper}
Suppose $\bm{p}=\frac{1}{n^{d-1}}(a_1,\cdots,a_{\kappa_d}) = \Omega\lp\frac{1}{n^{d-1}}\rp$ and $a_i\approx a_j \;\forall i,j=1,\ldots,\kappa_d $. If \eqref{eq:main_condi_main} and \eqref{eq:main_condi_order} holds as $n\ra\infty$. Then, the combined estimator $\what{\sigma}_2$ (Algorithm~\ref{alg:refine}) along with estimator $\what{\sigma}_1$ (Algorithm~\ref{alg:spec_init}) is able to achieve a risk of
\begin{equation}
\label{eq:theo_upper}
\sup_{(\mbf{B},\sigma)\in\Theta_d^0} \Risk_{\sigma}(\what{\sigma}_2) \leq \exp\bigg( -(1+\zeta_n) \sum_{i<j:(r_i,r_j)\in\mcal{N}_d} m_{r_i r_j} I_{p_i p_j} \bigg)
\end{equation}
for some vanishing sequence $\zeta_n \ra0$.
If $k$ is a constant, then \eqref{eq:theo_upper} holds without the further assumption \eqref{eq:main_condi_order}.
\end{theorem}


\noindent
In what follows, we first state the theoretical guarantees of Algorithm~\ref{alg:refine} as well as Algorithm~\ref{alg:spec_init} and demonstrate how they in combination aggregate to the upper bound result. The detailed proofs of the intermediate theorems are established later in Subsection~\ref{subsec:theo_refine} and Subsection~\ref{subsec:theo_specclus}, respectively.

The algorithm proposed in Section~\ref{sec:algo} consists of two steps. We first get a rough estimation through the first step, which is a spectral clustering on the hypergraph Laplacian matrix. After that, for each node we perform a local maximum likelihood estimation, which serves as the second step, to further adjust its community assignment. It turns out that this refining mechanism is actually crucial in achieving the optimal minimax risk $\Risk_d^{\ast}$, as long as the first initialization step satisfies a certain weak consistency condition. Specifically, the first-step algorithm $\what{\sigma}_0$ should meet the requirement stated below.

\begin{condition}
\label{con:1st_step}
There exists constant $C_0,\delta>0$, and a positive sequence $\gamma_n$ such that
\begin{equation}
\label{eq:condi_1st_step}
\inf_{(\mbf{B},\sigma)\in\Theta_d^0} \Pr_{\sigma}\lbp \loss(\what{\sigma}_0,\sigma)\leq \gamma_n \rbp \geq 1 - C_0n^{-(1+\delta)}
\end{equation}
for sufficiently large $n$.
\end{condition}

\noindent
We have the following performance guarantee for our second-step algorithm.
\begin{theorem}
\label{thm:theo_refine}
If
\begin{equation}
\label{eq:upper_condi}
\frac{\sum_{i<j:(r_i,r_j)\in\mcal{N}_d} m_{r_i r_j} I_{p_i p_j}}{\log k} \ra \infty
\end{equation}
and Condition~\ref{con:1st_step} is satisfied for
\begin{equation}
\label{eq:gamma_original}
\gamma = o\lp \frac{1}{k\log k} \rp
\end{equation}
and
\begin{equation}
\label{eq:gamma_further}
\gamma = o\bigg( \frac{1}{k \max_{(i,j):(r_i,r_j)\in\mcal{N}_d} \frac{p_1-p_{\kappa_d}}{p_i-p_j}} \bigg).
\end{equation}
Then, the estimator $\what{\sigma}_2$ (Algorithm~\ref{alg:refine}) is able to achieve a risk of
\begin{equation}
\label{eq:upper_result}
\sup_{(\mbf{B},\sigma)\in\Theta_d^0} \Risk_{\sigma}(\what{\sigma}_2) \leq \exp\bigg( -(1+\zeta_n^{\prime}) \sum_{i<j:(r_i,r_j)\in\mcal{N}_d} m_{r_i r_j} I_{p_i p_j} \bigg)
\end{equation}
for some vanishing sequence $\zeta_n^{\prime} \ra0$. If $k$ is a constant, then \eqref{eq:upper_result} holds without further assuming \eqref{eq:gamma_further}.
\end{theorem}

As for the initialization algorithm, first recall that we assume the connecting probabilities $\bm{p}=(p_1,\ldots,p_{\kappa_d})= \frac{1}{n^{d-1}}(a_1,\ldots,a_{\kappa_d})$ are in the same order. Also, we use $\lambda_k$ to denote the $k$-th largest singular value of the matrix $\E\mcal{L}(\mbf{A})$, which is the expectation of the hypergraph Laplacian \eqref{eq:hyper_Lapla}. Note that each entry in the matrix is a weighted combination of the probability parameters $p_i$'s. Stated in terms of $\lambda_k$, the following theorem characterizes the mismatch ratio of the first-step algorithm that we propose.

\begin{theorem}
\label{SPEC}
If
\begin{equation}
\label{assump22}
\frac{k a_1}{\lambda_k^2}\leq C_1
\end{equation}
for some sufficiently small $C_1\in(0,1)$ where $p_1 = \frac{a_1}{n^{d-1}}$. Apply Algorithm~\ref{alg:spec_init} with a sufficiently small constant $\mu>0$ and $\tau = C_2\bar{d}$ for some sufficiently large constant $C_2>0$. For any constant $C^{\prime}>0$, there exists some $C>0$ depending only on $C^{\prime},C_2$, and $\mu$ so that
\begin{equation*}
\loss(\what{\sigma},\sigma) \leq C\frac{a_1}{\lambda_k^2}
\end{equation*}
with probability at least $1-n^{-C^{\prime}}$.
\end{theorem}

\noindent
To take out the dependency on $\lambda_k$, we use the observation below.

\begin{lemma}
\label{HOMOEIG4anyD}
For $d$-{\hSBM} in $\Theta_d^0(n,k,\bm{p},\eta)$, we have
\begin{equation}\label{homoeigen4anyD}
\lambda_k \gtrsim \sum_{i<j:(r_i,r_j)\in\mcal{N}_d} m_{r_i r_j}(p_i-p_j).
\end{equation}
\end{lemma}

\begin{proof}[Proof of Theorem~\ref{thm:theo_upper}]
Finally, we only need to prove that the result of Theorem~\ref{SPEC} does match Condition~\ref{con:1st_step}.
Combining Theorem~\ref{SPEC} with Lemma~\ref{HOMOEIG4anyD}, we have
\begin{align*}
\loss(\what{\sigma},\sigma) &\leq \frac{Cn^{d-1}p_1}{(\lambda_k)^2} \\
 &\lesssim \frac{Cn^{d-1}p_1}{\sum_{i<j:(r_i,r_j)\in\mcal{N}_d} m_{r_i r_j}^2(p_i-p_j)^2} \\
 &\approx \frac{Cn^{d-1}}{\sum_{i<j:(r_i,r_j)\in\mcal{N}_d} m_{r_i r_j}^2I_{p_i p_j}}.
\end{align*}
Hence, we need $\frac{Cn^{d-1}}{\sum_{i<j:(r_i,r_j)\in\mcal{N}_d} m_{r_i r_j}^2I_{p_i p_j}} = o\lp\frac{1}{k\log k}\rp$, which means that
\begin{equation}
\label{DCONDI}
\frac{\sum_{i<j:(r_i,r_j)\in\mcal{N}_d} m_{r_i r_j}^2I_{p_i p_j}}{n^{d-1}k\log k} \ra\infty.
\end{equation}
Compared to the original requirement \eqref{eq:upper_condi} in the main theorem in Section~\ref{sec:main}, which is
\begin{equation*}
\frac{\sum_{i<j:(r_i,r_j)\in\mcal{N}_d} m_{r_i r_j}I_{p_i p_j}}{\log k} \ra \infty.
\end{equation*}
By counting arguments, we have
\begin{equation*}
\frac{1}{k^{d-1}} \leq \frac{m_{r_i r_j}}{n^{d-1}} \leq \frac{1}{k}.
\end{equation*}
We can see that the new requirement \eqref{DCONDI} is only slightly more stringent than the original one. Together, we complete the proof.
\end{proof}

\subsection{Refinement Step}
\label{subsec:theo_refine}



To prove Theorem~\ref{thm:theo_refine}, we need a couple of technical lemmas, the proofs of which are delegated to the appendices.
First, the lemma below ensures the accuracy of the parameter estimation with a qualified initialization algorithm.
\begin{lemma}
\label{lma:param_estim}
Suppose as $n\ra\infty$, $m_{r_i r_j}I_{p_i p_j}\ra\infty$ for each $(i,j)$ pair such that $(r_i,r_j)\in\mcal{N}_d$, and Condition~\ref{con:1st_step} holds with $\gamma$ satisfying \eqref{eq:gamma_original} for some $\delta>0$.
Then there exists a sequence $\zeta_n\ra0$ as $n\ra\infty$ and a constant $C>0$ such that
\begin{equation}
\adjustlimits\inf_{(\mbf{B},\sigma)\in\Theta_d^0}\min_{u\in[n]} \Pr_{\sigma}\Big\{\min_{\pi\in\mcal{S}_k}\max_{\bm{s} \in [k]^d} \big|\what{\mrm{B}}_{\bm{s}}^u-\mrm{B}_{\bm{s}_{\pi}}\big| \leq \zeta_n\max_{i<j:(r_i,r_j)\in\mcal{N}_d}(p_i-p_j) \Big\} \geq 1-Cn^{-(1+\delta)}.
\end{equation}
\end{lemma}
\noindent
Based on \autoref{lma:param_estim}, the next lemma shows that the local MLE method \eqref{eq:local_MLE} is able to achieve a risk that decays exponentially fast.

\begin{lemma}
\label{lma:localMLE_risk}
Suppose $m_{r_i r_j}I_{p_i p_j}\ra\infty$ as $n\ra\infty$ for each $(i,j)$ pair such that $(r_i,r_j)\in\mcal{N}_d$, and either $k$ is a constant or \eqref{eq:main_condi_order} holds. If there are two sequences $\gamma_n = o\lp\frac{1}{k}\rp$ and $\zeta_n^{\prime}=o(1)$, constants $C,\delta>0$, and permutations $\lbp\pi_u\rbp_{u=1}^n \subset \mcal{S}_k$ such that
\begin{equation*}
\adjustlimits\inf_{(\mbf{B},\sigma)\in\Theta_d^0}\min_{u\in[n]} \Pr_{\sigma}\Big\{ \loss_0((\what{\sigma}_u)_{\pi_u},\sigma)\leq\gamma_n, \ \big|\what{\mrm{B}}_{\bm{s}}^u-\mrm{B}_{\bm{s}_{\pi}}\big| \leq \zeta_n^{\prime} \max_{i<j:(r_i,r_j)\in\mcal{N}_d}(p_i-p_j) \Big\} \geq 1-Cn^{-(1+\delta)}.
\end{equation*}
Then for the local maximum likelihood estimator $\what{\sigma}_u(u)$ \eqref{eq:local_MLE}, there exists a vanishing sequence $\zeta_n^{''}\ra0$ such that
\begin{equation*}
\adjustlimits\sup_{(\mbf{B},\sigma)\in\Theta_d^0}\max_{u\in[n]} \Pr_{\sigma}\Big\{ (\what{\sigma}_u(u))_{\pi_u} \neq \sigma(u)\Big\} \leq (k-1)\cdot \exp\bigg( -(1-\zeta_n^{''}) \sum_{i<j:(r_i,r_j)\in\mcal{N}_d} m_{r_i r_j}I_{p_i p_j} \bigg) + Cn^{-(1+\delta)}.
\end{equation*}
\end{lemma}
\noindent
Finally, we justify the validity of using \eqref{eq:css_method} as a consensus majority voting with the following lemma.

\begin{lemma}[Lemma~4 in \cite{GaoMa_15}]
\label{lma:consensus}
For any labeling functions $\sigma$ and $\sigma^{\prime}$: $[n]\ra[k]$, if for some constant $C\geq1$,
\begin{equation*}
\min_{t\in[k]}\big|\lbp u\mid\sigma(u)=t \rbp\big| \geq \frac{n}{Ck}, \; \min_{t\in[k]}\big|\lbp u\mid\sigma^{\prime}(u)=t \rbp\big| \geq \frac{n}{Ck}, \text{ and } \min_{\pi\in\mcal{S}_k} \loss_0\big(\sigma^{\prime}_{\pi}\big) < \frac{1}{Ck}.
\end{equation*}
Define map $\xi:[k]\ra[k]$ as
\begin{equation*}
\xi(i) = \argmax_{t\in[k]} \big| \{u\mid\sigma(u)=t\} \cap \{u\mid\sigma^{\prime}(u)=t\} \big|
\end{equation*}
for each $i\in[k]$. Then $\xi\in\mcal{S}_k$ and $\loss_0(\xi(\sigma^{\prime}),\sigma)$ is equal to $\min_{\pi\in\mcal{S}_k}\loss_0(\sigma^{\prime}_{\pi},\sigma)$
\end{lemma}

We are ready to present the proof of Theorem~\ref{thm:theo_refine}.

\begin{proof}[Proof of Theorem~\ref{thm:theo_refine}]
Fix any $(\mbf{B},\sigma)$ in $\Theta_d^0$. Let $C_0$, $\delta>0$ be constants and $\gamma_n$ be a positive sequence in Condition~\ref{con:1st_step}. For each $u\in[n]$, there exists some $\pi_u\in\mcal{S}_k$ so that
\begin{equation*}
\Pr_{\sigma}\lbp \loss_{0}((\what{\sigma}_u)_{\pi_u},\sigma) \leq \gamma_n\rbp \geq 1 - C_0n^{-(1+\delta)}.
\end{equation*}
In consequence,
\begin{align*}
\E_{\sigma}\loss_0(\what{\sigma},\sigma) &= \E_{\sigma}\bigg[ \frac{1}{n}\sum_{u\in[n]}\Indc{(\what{\sigma}_u(u))_{\pi^{\mrm{CSS}}} \neq \sigma(u)} \bigg] \\
 &=\frac{1}{n}\sum_{u\in[n]} \Pr_{\sigma}\lbp (\what{\sigma}_{u}(u))_{\pi^{\mrm{CSS}}} \neq \sigma(u) \rbp \\
 &\leq \frac{1}{n}\sum_{u\in[n]} \Pr_{\sigma}\lbp (\what{\sigma}_{u}(u))_{\pi_{u}}\neq\sigma(u) \rbp + \Pr_{\sigma}\lbp \pi^{\mrm{CSS}}\neq\pi_{u} \rbp
\end{align*}
where $\pi^{\mrm{CSS}}$ is the consensus permutation \eqref{eq:css_method} in Algorithm~\ref{alg:refine}. By Lemma~\ref{lma:localMLE_risk}, for any $(\mbf{B},\sigma)\in\Theta_d^0$ and $u\in[n]$,
\begin{align*}
 &\Pr_{\sigma}\lbp (\what{\sigma}_u(u))_{\pi_u}\neq\sigma(u) \rbp \leq (k-1)\cdot \exp\bigg( -(1-\zeta_n^{''}) \sum_{i<j:(r_i,r_j)\in\mcal{N}_d} m_{r_i r_j}I_{p_i p_j} \bigg) + Cn^{-(1+\delta)}
\end{align*}
for some constants $C,\delta>0$ and $\zeta_n^{''}\ra0$. Moreover, \autoref{lma:consensus} implies that $\Pr\lbp \pi^{\mrm{CSS}}\neq\pi_u \rbp \leq C^{\prime}n^{-(1+\delta)}$. Together,
\begin{align*}
\sup_{(\mbf{B},\sigma)\in\Theta_d^0} \E_{\sigma}\loss_0(\what{\sigma},\sigma) \leq (k-1)\cdot \exp\bigg( -(1-\zeta_n^{''}) \sum_{i<j:(r_i,r_j)\in\mcal{N}_d} m_{r_i r_j}I_{p_i p_j} \bigg) + C^{''}n^{-(1+\delta)}.
\end{align*}
Since we assume that $\lim_{n\ra\infty}\sum_{i<j:(r_i,r_j)\in\mcal{N}_d} m_{r_i r_j}I_{p_i p_j} \ra\infty$, we further have
\begin{align*}
\sup_{(\mbf{B},\sigma)\in\Theta_d^0} \E_{\sigma}\loss_0(\what{\sigma},\sigma) &\leq \exp\bigg( -(1-\zeta_n^{'''}) \sum_{i<j:(r_i,r_j)\in\mcal{N}_d} m_{r_i r_j}I_{p_i p_j} \bigg) + C^{''}n^{-(1+\delta)} \\
 &= \circnum{1} + \circnum{2}
\end{align*}
for some $\zeta_n^{'''}\ra0$. If \circnum{1} decays faster than \circnum{2}, then $\Risk_d^{\ast} = o(\frac{1}{n^{1+\delta}}) < \frac{1}{n}$ for sufficiently large $n$. Therefore, $\Risk_d^{\ast} = 0$ and the corresponding parameters satisfy the criterion of exact recovery. On the other hand, if \circnum{1} dominates \circnum{2}, then there exists $\zeta_n \ra 0$ such that
\begin{equation*}
\Risk_d^{\ast} \leq \exp\bigg( -(1-\zeta_n)\sum_{i<j:(r_i,r_j)\in\mcal{N}_d} m_{r_i r_j}I_{p_i p_j} \bigg).
\end{equation*}
In either case, the claimed upper bound is achieved.
\end{proof}


The main difficulty of our proof above lies in proving \autoref{lma:param_estim} and \autoref{lma:localMLE_risk}. Compared to the graph {\SBM} case \cite{GaoMa_15}, we are dealing with more kinds of random variables and more kinds of relations. Since the original analysis is already exhausted, we believe that our generalization is not a trivial work.
\autoref{thm:theo_refine} implies that as long as we have a good initialization that achieves Condition~\ref{con:1st_step}, we could apply our refinement algorithm to make the risk decay exponentially fast at the desired rate.
This is because once Condition~\ref{con:1st_step} is satisfied, we could find the correct permutation by the consensus step and correctly estimate the parameters. Also, since Condition~\ref{con:1st_step} guarantees that we will only have $o(1)$ proportion of misclassified nodes, it is not hard to see that the local MLE method will work well.

\subsection{Spectral Clustering}
\label{subsec:theo_specclus}

Combined with Lemma~\ref{HOMOEIG4anyD}, we have the following corollary to Theorem~\ref{SPEC} in terms of $m_{r_i r_j}$ and $I_{p_i p_j}$ only.

\begin{corollary}
\label{cor:theo_spectral}
Suppose $\bm{p}=\frac{1}{n^{d-1}}(a_1,\cdots,a_{\kappa_d})=\Omega\lp\frac{1}{n^{d-1}}\rp$ and $a_i\approx a_j \;\forall i,j=1,\ldots,\kappa_d$. If
\begin{equation}
\label{eq:spec_condi}
\frac{\sum_{i<j:(r_i,r_j)\in\mcal{N}_d} m_{r_i r_j} I_{p_i p_j}}{k^d} \geq \frac{1}{c}
\end{equation}
for some sufficiently small $c\in(0,1)$. Then, with the estimator $\what{\sigma}_1$ (Algorithm~\ref{alg:spec_init}), corresponding to any constant $C^{\prime}>0$ there exists a constant $C=C(C^{\prime})>0$ such that
\begin{equation}
\label{eq:spec_result}
\sup_{(\mbf{B},\sigma)\in\Theta_d^0} \Pr_{\sigma}\bigg\{ \loss(\what{\sigma}_1,\sigma) \leq C\frac{k^{d-1}}{\sum_{i<j:(r_i,r_j)\in\mcal{N}_d} m_{r_i r_j} I_{p_i p_j}} \bigg\} \geq 1 - n^{-C^{\prime}}.
\end{equation}
\end{corollary}

\subsubsection{Comparison with \cite{GhoshdastidarDukkipati_17}}
\label{subsubsec:theo_comp}
In \cite{GhoshdastidarDukkipati_17}, the authors consider the "balanced partitions in uniform hypergraph" as a special case. In order to have a fair comparison, we would focus on the performance on this sub-parameter space. Essentially, it is the homogeneous and approximatedly equal-sized $\Theta_d^0$ that we consider, except that the authors only distinguish out the \emph{all-same} community relation (associated with $\Ber(p)$) from the rest of all possible community relations in $\mcal{N}_d$ (associated with $\Ber(q)$). We denote this parameter space as $\Theta_d^s$. The consistency result derived for $\Theta_d^s$ with the spectral hypergraph partition algorithm $\what{\sigma}_{\mrm{GD}}$ proposed therein can be summarized as follows: If
\begin{equation}
\label{eq:GD_condi}
p \geq C\frac{k^{2d-1}(\log n)^2}{n^{d-1}}
\end{equation}
for some absolute constant $C>0$. Then,
\begin{equation}
\label{eq:GD_result}
\sup_{(\mbf{B},\sigma)\in\Theta_d^s}\Pr_{\sigma}\bigg\{ \Risk_{\sigma}(\what{\sigma}_{\mrm{GD}}) = O\lp \frac{1}{k\log n}\rp \bigg\} \geq 1 - O\lp (\log n)^{-\frac{1}{4}}\rp.
\end{equation}
\noindent
On the other hand, the theoretical guarantee \autoref{cor:theo_spectral} for the spectral clustering algorithm $\what{\sigma}_1$ specializes to
\begin{corollary}
\label{cor:1st_step_comp}
Suppose $\bm{p}=(p,q,\ldots,q)$ and $p\approx q$. If
\begin{equation}
\label{eq:comp_condi}
p \geq \frac{1}{n^{d-1}}\cdot C\frac{k^{2d+1}}{n^2}
\end{equation}
for some absolute constant $C>0$. Then,
\begin{equation}
\label{eq:comp_result_diff}
\sup_{(\mbf{B},\sigma)\in\Theta_d^s}\Pr_{\sigma}\bigg\{ \Risk_{\sigma}(\what{\sigma}_1) = O\Big( \frac{1}{k}\Big) \bigg\} \geq 1 - n^{-O(1)}
\end{equation}
\end{corollary}
\noindent
under $\Theta_d^s$ with some manipulations that are similar to the derivation of \autoref{HOMOEIG4anyD}. It turns out that we allow the observed hypergraph to be sparser (a lower connecting probability) yet acquire a higher average mismatch ratio compared to the result obtained in \cite{GhoshdastidarDukkipati_17}. However, if we raise the probability parameters $\bm{p}=(p,q,\ldots,q)$ to the same order as in \eqref{eq:GD_condi}, the risk obtained in \eqref{eq:comp_result_diff} will become
\begin{equation}
\label{eq:comp_result_same}
\sup_{(\mbf{B},\sigma)\in\Theta_d^s}\Pr_{\sigma}\lbp \Risk_{\sigma}(\what{\sigma}_1) = O\lp \frac{k}{n^2 \log n}\rp \rbp \geq 1 - n^{-O(1)}.
\end{equation}
That is, a $\frac{k^2}{n^2} = \frac{1}{(n^{\prime})^2}$ gain in terms of the mismatch ratio over \eqref{eq:GD_result}. In either case, we always have a success probability that converges faster to $1$ for the weak consistency performance guarantee.

\subsubsection{Proof of \autoref{SPEC}}
\label{subsubsec:theo_spec_clus_proof}
The proof mainly follows from \cite{GaoMa_15} and the main differences are the lemmas where we extend them for the d-hSBM case. For the sake of completeness, we include the extended proof below.

The proof requires the following two lemmas, the details of which are defered to the appendices.
First, we demonstrate that the trimmed hypergraph Laplacian does not deviate too much from its untrimmed expectation.

\begin{lemma}
\label{SPECLEMMADHSBM}
$\forall C^{\prime}>0$, $\exists \ C>0$ such that
\begin{equation*}
\lnorm \msf{T}_{\tau}(\mcal{L}(\mbf{A}))-\E\mcal{L}(\mbf{A}) \rnorm_{\mrm{op}}\leq C\sqrt{n^{d-1}p_1+1}
\end{equation*}
with probability at least $1-n^{-C^{\prime}}$ uniformly over $\tau\in \lb C_1(n^{d-1}p_1+1),C_2(n^{d-1}p_1+1)\rb$ for some sufficiently large constants $C_1$ and $C_2$.
\end{lemma}

\noindent
The next lemma analyzes the spectrum of $\E\mcal{L}(\mbf{A})$ and pinpoints a special structure.

\begin{lemma}[Lemma~6 in \cite{GaoMa_15}]
\label{SPECLEMMA6}
We have
\begin{equation*}
\msf{SVD}_k\lp\E\mcal{L}(\mbf{A})\rp = \mbf{U}\mbf{\Lambda}\mbf{U}^\mrm{T}
\end{equation*}
where $\mbf{U}=\mbf{Z}\mbf{\Delta}^{-1}\mbf{W}$ with $\mbf{\Delta} = \msf{diag}(\sqrt{n_1},\ldots,\sqrt{n_k})$. $\mbf{Z}\in\{0,1\}^{n\by k}$ is a matrix with exactly one nonzero entry in each row at $(i,\sigma(i))$ taking value $1$ and $\mbf{W}\in\mrm{O}(k,k)$.
\end{lemma}

\begin{proof}[Proof of Theorem~\ref{SPEC}]
Under the assumption $p_1\approx p_{\kappa_d}$, we have $\E\tau\in[C_1'n^{d-1}p_1,C_2'n^{d-1}p_1 ]$ for some large constant $C_1', C_2'$. Thus by Bernstein's inequality, with probability at least $1-e^{-C'n}$, we have $\tau \in [C_1n^{d-1}p_1,C_2n^{d-1}p_1 ] $. By Davis-Kahan theorem \cite{DavisKahan_70}, we have
\begin{equation*}
||\what{\mbf{U}} - \mbf{U}\mbf{W_1} ||_{\mrm{F}} \leq C\frac{\sqrt{k}}{\lambda_k}||\msf{T}_{\tau}(\mcal{L}(\mbf{A}))-\E\mcal{L}(\mbf{A})||_{\mrm{op}}
\end{equation*}
for some $\mbf{W_1} \in O(k,k)$ and some constant $C>0$. Then applying \autoref{SPECLEMMA6}, we have
\begin{equation}
\label{60}
||\what{\mbf{U}} - \mbf{V} ||_{\mrm{F}} \leq C\frac{\sqrt{k}}{\lambda_k}||\msf{T}_{\tau}(\mcal{L}(\mbf{A}))-\E\mcal{L}(\mbf{A})||_{\mrm{op}}
\end{equation}
where $\mbf{V} = \mbf{Z}\mbf{\Delta}^{-1}\mbf{W} = [\bm{v}_1^{\mrm{T}}...\bm{v}_n^{\mrm{T}}]^{\mrm{T}}$ as we state in \autoref{SPECLEMMA6}. Combining \eqref{60}, \autoref{SPECLEMMADHSBM} and the conclusion $\tau \in [C_1n^{d-1}p_1,C_2n^{d-1}p_1 ]$ with probability at least $1-e^{-C'n}$, we have
\begin{equation}
\label{61}
||\what{\mbf{U}} - \mbf{V}||_{\mrm{F}} \leq\frac{C \sqrt{k}\sqrt{n^{d-1}p_1}}{\lambda_k}
\end{equation}
with probability at least $1-n^{-C'}$. The definition of $\mbf{V}$ implies that
\begin{equation*}
\lnorm \bm{v}_i - \bm{v}_j \rnorm_2
\begin{cases}
\geq \sqrt{\frac{2k}{n}}, & \mbox{, if } \sigma(i)\neq\sigma(j) \\
=0, & \mbox{, otherwise}.
\end{cases}
\end{equation*}
Let $\mbf{X} = \mbf{\Delta}^{-1}\mbf{W}$, which means $\bm{v}_i=\bm{x}_{\sigma(i)}$. Recall the definition of critical radius $r = \nu \sqrt{\frac{k}{n}}$ in Algorithm~\ref{alg:spec_init}. Define the sets
\begin{equation}
\label{Tset}
T_i = \lbp s \in \sigma^{-1}(i) : \lnorm\what{\bm{u}}_s - \mbf{x}_i\rnorm_2 < \frac{r}{2} \rbp, \;i\in[k].
\end{equation}
By definition, $T_i \cap T_j = \varnothing$ for $i\neq j$ and
\begin{equation*}
\bigcup_{i\in[k]}T_i = \lbp s\in[n] : \lnorm\what{\bm{u}}_s - \bm{v}_s\rnorm_2 < \frac{r}{2} \rbp.
\end{equation*}
Thus,
\begin{equation*}
\bigg|\Big( \bigcup_{i\in[k]}T_i\Big)^c\bigg| \cdot \frac{r^2}{4}\leq \sum_{s\in[n]} \lnorm\what{\bm{u}}_s - \bm{v}_s\rnorm_2^2 \leq \frac{C^2kn^{d-1}p_1}{\lambda_k^2}
\end{equation*}
where the last inequality is due to \eqref{61}. The rest of the proof is identical to the proof of Theorem $3$ in \cite{GaoMa_15}. 
For completeness, we shall repeat it again here. After some rearrangements, we have
\begin{equation}
\label{eq64}
\bigg|\Big( \bigcup_{i\in[k]}T_i\Big)^c\bigg| \leq \frac{4C^2n^{d}p_1}{\mu^2\lambda_k^2}.
\end{equation}
It means most nodes are close to the centers and are in the set we defined in \eqref{Tset}. Also note that the sets $\{T_i\}_{i\in[k]}$ are disjoint. Suppose there is some $i\in [k]$ such that $|T_i| < |\sigma^{-1}(i)|-|(\cup_{i\in[k]}T_i)^c|$, we have $|\cup_{i\in[k]}T_i| = \sum_{i\in[k]}|T_i| < n-|(\cup_{i\in[k]}T_i)^c| = |\cup_{i\in[k]}T_i\|$, which leads to a contradiction. Thus, we must have
\begin{equation*}
|T_i| \leq |\sigma^{-1}(i)|-\bigg|\Big( \bigcup_{i\in[k]}T_i\Big)^c\bigg| \leq \frac{n}{k} - \frac{4C^2n^{d}p_1}{\mu^2\lambda_k^2} > \frac{n}{2k}
\end{equation*}
where the last inequality is from the assumption \eqref{assump22}. Since the cluster centers are at least $\sqrt{\frac{2k}{n}}$ apart from each others and both $\{T_i\}_{i\in[k]}$, $\{\what{C}_i\}_{i\in[k]}$ (recall that $\what{C}_i$ are defined in Algorithm~\ref{alg:spec_init}) are defined through the critical radius $r = \mu \sqrt{\frac{k}{n}}$, each $\what{C}_i$ should intersect with only one $T_i$. We claim that there is a permutation $\pi$ of set $[k]$, such that
\begin{equation}
\label{claim66}
\what{C}_i\bigcap T_{\pi(i)} \neq \varnothing,\; \lba\what{C}_i\rba \geq \lba T_{\pi(i)}\rba \forall i\in[k].
\end{equation}
We could now continue the proof with claim \eqref{claim66}, where the proof of \eqref{claim66} can be found in \cite{GaoMa_15} (in their proof of Theorem~3). It is mainly established by an easy mathematical induction. From the definition of $\what{C}_i$ and \eqref{claim66}, we have for any $i\neq j$, $T_{\pi(i)}\cap\what{C}_j = \varnothing$. This directly implies that for any $i\neq j$, $T_{\pi(i)}\subset\what{C}_j$. Thus, we know that $T_{\pi(i)}\cap\what{C}_i^c \subset \lp\cup_{i\in[k]}\what{C}_i \rp^c$. Therefore,
\begin{equation*}
\bigcup_{i\in[k]}\Big(T_{\pi(i)} \bigcap \widehat{C}_i^c \Big) \subset \Big(\bigcup_{i\in[k]}\widehat{C}_i \Big)^c.
\end{equation*}
Combining with the fact that $T_i \cap T_j = \varnothing \; \forall i\neq j$, we have
\begin{equation}
\label{eq68}
\sum_{i\in[k]}\Big| T_{\pi(i)} \bigcap \widehat{C}_i^c \Big| \leq \Big| \Big(\bigcup_{i\in[k]}\widehat{C}_i \Big)^c \Big|.
\end{equation}
By definition, $\widehat{C}_i \bigcap \widehat{C}_j = \varnothing \; \forall i\neq j$. Along with \eqref{claim66}, we have
\begin{equation}
\label{eq69}
\bigg| \Big(\bigcup_{i\in[k]}\widehat{C}_i \Big)^c \bigg| = n - \sum_{i\in[k]} \lba \widehat{C}_i \rba \leq n - \sum_{i\in[k]} \lba T_i \rba = \bigg| \Big(\bigcup_{i\in[k]}T_i \Big)^c \bigg|.
\end{equation}
Together with \eqref{eq64}, \eqref{eq68} and \eqref{eq69}, we have
\begin{equation}
\label{eq70}
\sum_{i\in[k]}\Big| T_{\pi(i)} \bigcap \widehat{C}_i^c \Big| \leq \frac{4C^2n^{d}p_1}{\mu^2\lambda_k^2}
\end{equation}
Since for each $u \in \cup_{i\in[k]}\lp T_{\pi(i)} \cap \widehat{C}_i \rp$, we have $\widehat{\sigma}(u) = i$ when $\sigma(u) = \pi(i)$, the mis-classification ratio is bounded by
\begin{align*}
\loss_0(\widehat{\sigma}, \pi^{-1}(\sigma)) &\leq \frac{1}{n} \bigg| \Big(\bigcup_{i\in[k]}\big( T_{\pi(i)} \bigcap \widehat{C}_i \big) \Big)^c\bigg| \\
 & \leq \frac{1}{n} \lb \bigg| \Big(\bigcup_{i\in[k]}\big( T_{\pi(i)} \bigcap \widehat{C}_i \big) \Big)^c \bigcap \Big(\bigcup_{i\in[k]}T_i \Big) \bigg| + \bigg| \Big(\bigcup_{i\in[k]}T_i \Big)^c \bigg| \rb \\
 & \leq \frac{1}{n} \lb \sum_{i\in[k]} \bigg| T_{\pi(i)} \bigcap \widehat{C}_i^c \bigg| + \bigg| \Big(\bigcup_{i\in[k]}T_i \Big)^c  \bigg| \rb \\
 & \leq \frac{8C^2n^{d-1}p_1}{\mu^2\lambda_k^2}
\end{align*}
where the last inequality is from \eqref{eq64} and \eqref{eq70}. This proves the desired conclusion.
\end{proof}

\begin{remark}
Essentially, \autoref{SPEC} says that the performance of Algorithm~\ref{alg:spec_init} will be upper-bounded in regard to the $k$-th largest singular value. When $\lambda_k$ is large, it means that the singular vectors are well separated, ensuring the algorithm to have a good performance. This is similar to classical spectral clustering methods.
\end{remark} 

\section{Minimax Lower Bound}
\label{sec:conv}
By constructing a smaller parameter space where we can analyze the risk, we are able to establish the converse part of the main theorem as follows.
\begin{theorem}
\label{thm:lower}
If
\begin{equation}
\label{eq:lower_condi}
\sum_{i<j:(r_i,r_j)\in\mcal{N}_d} m_{r_i r_j} I_{p_i p_j} \ra \infty.
\end{equation}
Then
\begin{equation}
\label{eq:lower_result}
\adjustlimits\inf_{\what{\sigma}}\sup_{(\mbf{B},\sigma)\in\Theta_d^0} \Risk_{\sigma}(\what{\sigma}) \geq \exp\bigg( -(1-\zeta_n) \sum_{i<j:(r_i,r_j)\in\mcal{N}_d} m_{r_i r_j} I_{p_i p_j} \bigg)
\end{equation}
for some vanishing sequence $\zeta_n \ra0$.
\end{theorem}

\noindent
We can see that condition \eqref{eq:lower_condi} required here in the lower bound is less stringent than the condition \eqref{eq:main_condi_main} required in the upper bound.

To prove Theorem~\ref{thm:lower}, we first introduce the concept of local loss. The equivalence class of a community assignment $\sigma$ is defined as $\Gamma(\sigma) \eqDef \lbp \sigma^{\prime} \mid \exists \pi\in\mcal{S}_k \st \sigma^{\prime}=\sigma_{\pi} \rbp$. Let $S_{\sigma}(\what{\sigma}) = \lbp \sigma^{\prime} \in \Gamma(\what{\sigma}) \mid d_\mrm{H}(\sigma^{\prime},\sigma) = d_\mrm{H}(\what{\sigma},\sigma) \rbp$ be the set of all permutations in the equivalence class of $\what{\sigma}$ that achieve the minimum distance. For each $i\in[n]$, the local loss function is defined as the proportion of false labeling of node $i$ in $S_{\sigma}(\what{\sigma})$.
\begin{equation*}
\loss(\what{\sigma}(i),\sigma(i)) \eqDef \sum_{\sigma^{\prime} \in S_{\sigma}(\what{\sigma})} \frac{\Indc{\what{\sigma}(i) \neq \sigma(i)}}{\lvert S_{\sigma}(\what{\sigma}) \rvert}
\end{equation*}
It turns out that it is rather easy to study the local loss. In the proof of our converse, we consider a sub-parameter space of $\Theta_d^0$, i.e. $\Theta_d^L$ \eqref{eq:ThetaL} as defined in \autoref{sec:main}. Recall that the parameter space $\Theta_d^L$ is similar to $\Theta_d^0$, but we only allow each community size to be within $\lfloor\frac{n}{k}\rfloor\pm1$.
Since $\Theta_d^L$ is closed under permutation, we can apply the \emph{global-to-local} lemma in \cite{ZhangZhou_16}.

\begin{lemma}[Lemma~2.1 in \cite{ZhangZhou_16}]
\label{lma:global2local}
Let $\Theta$ be any homogeneous parameter space that is closed under permutation. Let $\Unif$ be the uniform prior over all the elements in $\Theta$. Define the global Bayesian risk as $\Risk_{\sigma\sim\Unif}(\what{\sigma}) = \frac{1}{|\Theta|} \sum_{\sigma\in\Theta} \E_{\sigma} \loss(\what{\sigma},\sigma)$ and the local Bayesian risk $\Risk_{\sigma\sim\Unif}(\what{\sigma}(1)) = \frac{1}{|\Theta|} \sum_{\sigma\in\Theta} \E_{\sigma} \loss(\what{\sigma}(1),\sigma(1))$ for the first node. Then
\begin{equation*}
\inf_{\what{\sigma}} \Risk_{\sigma\sim\Unif}(\what{\sigma}) = \inf_{\what{\sigma}} \Risk_{\sigma\sim\Unif}(\what{\sigma}(1)).
\end{equation*}
\end{lemma}

\noindent
Second, the local Bayesian risk can be transformed into the risk function of a hypothesis testing problem. We consider the most indistinguishable case where the potential candidate only disagrees with the ground truth on a single node. The key observation is that the situation is exactly the same as our testing one node at a time in the local version of the MLE method.

\begin{lemma}
\label{lma:hypo_test}
\begin{equation*}
\Risk_{\sigma\sim\Unif} (\what{\sigma}(1)) \geq \Pr\Bigg\{ \adjustlimits\sum_{i<j:(r_i,r_j)\in\mcal{N}_d}\sum_{u=1}^{m_{r_i r_j}} C_{r_i r_j}\lp X_u^{(r_j)}-X_u^{(r_i)}\rp \geq 0\Bigg\}
\end{equation*}
where $f(s) \eqDef \frac{s}{1-s}$ for $C_{xy} \eqDef \log\frac{f(x)}{f(y)}$, and for each $(r_i,r_j)$ pair in $\mcal{N}_d$, $X_u^{(r_j)}\diid\Ber(p_j), X_u^{(r_i)}\diid\Ber(p_i) \;\forall u=1,\ldots,m_{r_i r_j}$ are all mutually indepedent random variables.
\end{lemma}

\noindent
With the aid of the Rozovsky lower bound \cite{Rozovsky_03}, we are able to prove the following auxiliary result.

\begin{lemma}
\label{lma:rozovsky}
If $\sum_{i<j:(r_i,r_j)\in\mcal{N}_d}m_{r_i r_j}I_{p_i p_j} \ra\infty$, then
\begin{equation}
\label{eq:lmaroz}
\Pr\Bigg\{ \adjustlimits\sum_{i<j:(r_i,r_j)\in\mcal{N}_d}\sum_{u=1}^{m_{r_i r_j}} C_{r_i r_j}\lp X_u^{(r_j)}-X_u^{(r_i)}\rp \geq 0\Bigg\} \geq \exp\bigg( -(1+o(1)) \sum_{i<j:(r_i,r_j)\in\mcal{N}_d} m_{r_i r_j}I_{p_i p_j} \bigg).
\end{equation}
\end{lemma}

\begin{proof}[Proof of Theorem~\ref{thm:lower}]
Finally, since the Bayesian risk always lower bound the minimax risk, we have
\begin{align*}
\Risk_d^{\ast} &= \adjustlimits \inf_{\what{\sigma}}\sup_{\sigma\in\Theta_d^0} \E_{\sigma}\loss(\what{\sigma},\sigma) \geq \adjustlimits \inf_{\what{\sigma}}\sup_{\sigma\in\Theta_d^L} \E_{\sigma}\loss(\what{\sigma},\sigma) \\
 &\geq \inf_{\what{\sigma}} \Risk_{\sigma\sim\Unif}(\what{\sigma}) = \inf_{\what{\sigma}} \Risk_{\sigma\sim\Unif}(\what{\sigma}(1)) \\
 &\geq \Pr\Bigg\{ \adjustlimits\sum_{i<j:(r_i,r_j)\in\mcal{N}_d}\sum_{u=1}^{m_{r_i r_j}} C_{r_i r_j}\lp X_u^{(r_j)}-X_u^{(r_i)}\rp \geq 0\Bigg\} \\
 &\geq \exp\bigg( -(1+o(1))\sum_{i<j:(r_i,r_j)\in\mcal{N}_d} m_{r_i r_j}I_{p_i p_j} \bigg).
\end{align*}
\end{proof}

\section{Experimental Results}
\label{sec:expe}
The advantage of clustering with a hypergraph representation over traditional graph-based approaches has been reported in the literature \cite{AgarwalLim_05}, \cite{ZhouHuang_06}, \cite{GhoshdastidarDukkipati_15}, \cite{GhoshdastidarDukkipati_17}. Here, we present a comparative study of our two-step algorithm on generative $3$-{\hSBM} data with existing method, which manifests that our proposed algorithm indeed has a better performance
and that the second refinement step is actually crucial in achieving a lower risk.
In the following summary statistics, \textbf{Algo~2} refers to our first-step spectral clustering (i.e. Algorithm~\ref{alg:spec_init}) and \textbf{Algo~1} refers to the combined two-step workflow (i.e. Algorithm~\ref{alg:refine} on top of Algorithm~\ref{alg:spec_init}). We separate them apart to further see how much it can be improved on the mismatch ratio by using the local refinement mechanism.

We compare the performance of our Algorithm~\ref{alg:spec_init} and Algorithm~\ref{alg:refine} with the spectral non-uniform hypergraph partitioning (SnHP) method \cite{GhoshdastidarDukkipati_17} using the hypergraph Laplacian proposed in \cite{ZhouHuang_06} on generative $3$-{\hSBM} data. The parameter space tested is homogeneous and exactly equal-sized, which means that each community has the same number of members. This "nodes per community" parameter $n^{\prime}=n/k$ scales from $20,30,\cdots$ to $100$, while the number of communities $k$ varies from $2,3,\cdots$ to $10$. We set the connecting probability parameter $\bm{p}$ to be $(60,30,10)\times\log(n)/n^2$ for each possible value of $n=k\times n^{\prime}$. Note that the order of $\bm{p}$ is as prescribed for the sparse regime discussed in Section~\ref{sec:theo}. The choice of this particular triplet is to ensure that the generated hypergraphs are not too sparse to be connected. Specifically, the total number of realized hyperedges are roughly in the order of $4n\log(n)$ to $5 n\log(n)$. The performance under each scenario, i.e. each pair of $(k,n^{\prime})$, is averaged over $25$ random realizations. Figure~\ref{fig:syn} below summarizes our simulation results.

\begin{figure}[ht]
\centering
\subfigure[$k=3$]{
    \includegraphics[width=0.85\textwidth]{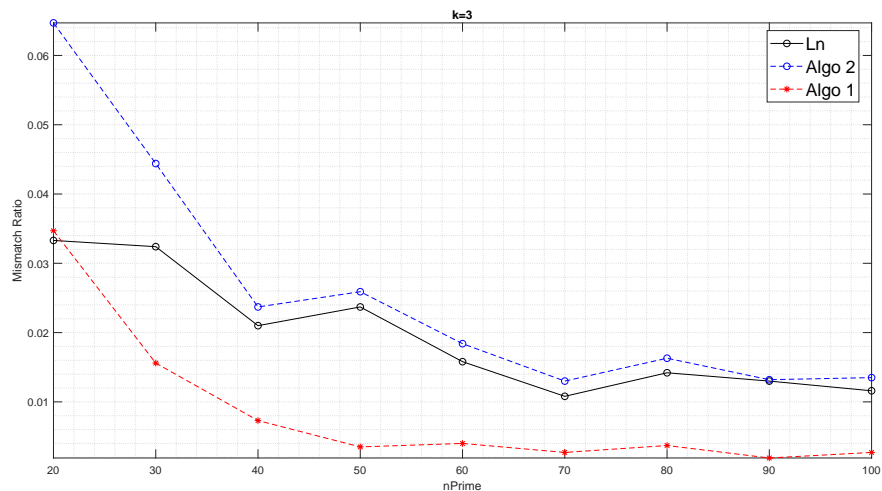}
}
\vspace{0pt}
\subfigure[$k=6$]{
    \includegraphics[width=0.85\textwidth]{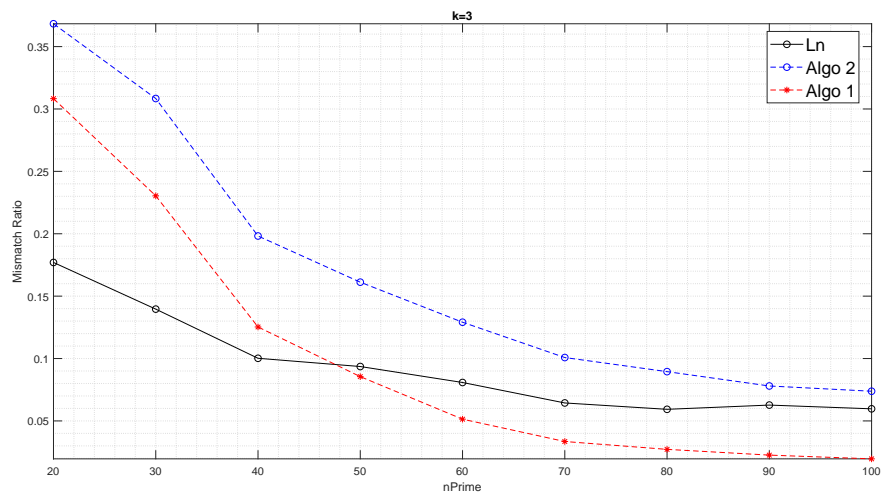}
}
\caption{Performance of Algo~2 and Algo~1 Compared to the Spectral Method \cite{GhoshdastidarDukkipati_17}.}
\label{fig:syn}
\end{figure}


Except for the first few scenarios where the total number of nodes $n$ are quite small, we can see that our spectral clustering algorithm performs roughly as well as the algorithm in \cite{GhoshdastidarDukkipati_17}. This somewhat indicates that the weak consistency condition Condition~\ref{con:1st_step} can also be satisfied by using the hypergraph Laplacian proposed by \cite{ZhouHuang_06} as the first initialization step.
Furthermore, the refinement scheme indeed has a better performance of the spectral clustering method in terms of the mismatch ratio. Observe that the improvement due to the second step becomes larger as $k$ (and hence $n$) increases.

\section{Discussions}
\label{sec:disc}


The idea of using the hypergraph Laplacian matrix in the clustering problem can be traced back to \cite{ZhouHuang_06}. However, 
statistical guarantee of its performance remains open. Theoretical guarantees of our hypergraph clustering method developed in \autoref{subsec:theo_specclus} can be viewed as an answer to this open question.

The proposed hypergraph clustering method basically encodes the group-wise interactions in an effective weighted graph, where the weight of each edge records the number of hyperedges involving the two nodes. This is very similar to the Labeled {\SBM} model studied in previous works \cite{YunProutiere_16}, where each edge takes on a nonnegative weight (label) in a finite field and the appearance of each edge is independent to the rest of the graph. 
The key difference between our algorithm and those in Labeled {\SBM} hence lies in the refinement step: for our $d$-{\hSBM} model, values of weight on the edges of the effective weighted graph are not mutually independent, as they are compressed results from higher-order interactions. The independence assumption in Labeled {\SBM} may not be practical in some applications. When viewed as the level of participation in various group interactions among the network, the weight between two entities may well depend on a hidden third party. 
Our approach can thus serve as a solution to that problem, since we directly treat the group-wise interactions in a higher-order form in the local refinement step. 

After exploring the similarities and differences between Labeled {\SBM} and hypergraph {\SBM}, an even more complete treatment of the problem of community detection would be a generalization of $d$-{\hSBM} to a weighted (or labeled) version. In either the Bayesian framework \cite{YunProutiere_16} or the minimax setting \cite{JogLoh_15}, it turns out that the R\'{e}nyi divergence of order $1/2$ still appears in the characterization of the threshold behavior of exact recovery. Extending from the divergence between two simple Bernoulli distributions, in a labeled observation of the network the R\'{e}nyi divergence becomes
\begin{equation}
\label{eq:Renyi_div_general}
I_{pq}^{\ell} =
\begin{cases}
-2\log\lp \int_{-\infty}^{\infty} \sqrt{p_n(x)q_n(x)}dx \rp & \text{, for continuous distribution on } \mbb{R} \\
-2\log\lp \sum_{\ell\geq0} \sqrt{p_n(\ell)q_n(\ell)} \rp & \text{, for discrete distribution on } \mbb{N}
\end{cases}
\end{equation}
between two more general edge weight distributions $p_n(x)$ and $q_n(x)$. Note that \eqref{eq:Renyi_div_Ber} is a special case of \eqref{eq:Renyi_div_general} when $p_n(x)=\Ber(p)$ and $q_n(x)=\Ber(q)$. We envision that $I_{pq}^{\ell}$ would play a major role characterizing the minimax risk in the extended labeled $d$-{\hSBM} model as well, and leave it for a possible direction of future work. Specifically, we conjecture that the minimax risk in such a labeled hypergraph probabilistic model would also be an exponentially decaying risk and the error exponent would be in the form of
\begin{equation*}
\sum_{i<j:(r_i,r_j)\in\mcal{N}_d} m_{r_i r_j}I_{p_i p_j}^{\ell}
\end{equation*}
where $I_{p_i p_j}^{\ell}$ is the R\'{e}nyi divergence of order $1/2$ between two hyperedge weight distributions $p_i=p_i(x)$ and $p_j=p_j(x)$.


Finally, we would like to comment on the extendability of the two-step algorithm and our proof techniques. 
The refine-after-initialization methodology is introduced in \cite{GaoMa_15} to achieve the minimax risk. 
We generalize this idea to the hypergraph setting and reach the conclusion that the minimax risk in $d$-{\hSBM} is also an exponential rate. Besides, the exponent of the minimax risk consists of terms of pairwise comparison which are one node different. This can be directly identified with the case that is most difficult to recover where there is only one node mis-classified. The matching of the form of the local MLE risk to the form derived in the converse is the key for proving optimality. 
The robustness of our two-step algorithm lies in the fact that it is able to achieve the minimax risk under any probabilistic model that assumes the independence of each ``group-wise'' interaction. The estimator of the parameters is adapted to the hypothesized probabilistic model and the local likelihood function is adjusted accordingly. 
Hence, with a different underlying probabilistic model, we believe the two-step algorithm and our proofs will still work, as long as the refinement step is adjusted according to the model.





\bibliographystyle{IEEEtran}
\bibliography{Ref}

\clearpage
\appendices




\section{Proof of Lemma~\ref{lma:param_estim}}
\label{app:lemma4param_estim}
Fix any $(\mbf{B},\sigma)\in\Theta_d^{0}$ and $u\in[n]$. We denote the induced community structure on the $n$ nodes as $[n]=\cup_{i=1}^k C_i$ where $C_i = \{v\in[n]\mid\sigma(v)=i\}$ is the $i$-th comomunity.
Define the event
\begin{equation}
\label{eq:good_init}
E_u \eqDef \big\{ \loss_0((\what{\sigma}_u)_{\pi_u},\sigma)\leq \gamma \big\}
\end{equation}
For simplicity, we assume that $\pi_u$ is the identity permutation. Fix any $i\in[k]$. Then, on $E_u$ we have
\begin{equation}
\label{Comcond}
n_i \geq \big|\wtild{C}_i^u\cap C_i\big| \geq n_i-\gamma_1n, \ \big|\wtild{C}_i^u\cap C_i\big|\leq \gamma_2n
\end{equation}
where $\gamma_1,\gamma_2 \geq 0$ and $\gamma_1+\gamma_2 \leq \gamma$.
Let $C_i^{\prime}$ be a deterministic subset of $[n]$ such that \eqref{Comcond} holds with $\wtild{C}_i^u$ replaced by $C_i^{\prime}$. By definition, there are at most
\begin{equation}\label{UNION}
   \sum_{l=0}^{\gamma n}\binom{n_i}{l}\sum_{m=0}^{\gamma n}\binom{n-n_i}{m}\leq \exp\lp C_1\gamma n \log \frac{1}{\gamma}\rp
\end{equation}
different subsets with this property for some absolute constant $C_1>0$. In the following, we will go through the case $\big|\what{\mrm{B}}_{i\cdot\mbf{1}}^u - \mrm{B}_{i\cdot\mbf{1}}\big| \leq o(\max_{(i,j):(r_i,r_j)\in \mcal{N}_d}p_i-p_j)$
where $i\cdot\mbf{1} \eqDef (i,i,\ldots,i)$ corresponds to the all-community-$i$ connection.
For the rest of the cases, we can easily follow an similar procedure to obtain the desired upper bound.

Let $\xi_{i}^{\prime}$ be the edges within $C_i^{\prime}$. Note that $\xi_{i}^{\prime}$ consists of $\binom{n_i}{d}$ independent Bernoulli random variables. The number of truly $\Ber(\mrm{B}_{i\cdot\mbf{1}})$'s is at least $\binom{n_i-\gamma n}{d}$. By an simple combinatorial argument, we have
\begin{equation}
\label{eq2}
\begin{aligned}
\E\Bigg[\frac{\lba\xi_{i}^{\prime}\rba}{\binom{\lba C_i^{\prime}\rba}{d}}\Bigg] \geq \min_{t\in[0,\gamma k]} \Bigg\{ p_i-(1-(1-t)^d)(p_i-p_{\mcal{K}_{d}})\Bigg\}
\end{aligned}
\end{equation}
\begin{equation}
\label{eq3}
\begin{aligned}
\E\Bigg[\frac{\lba\xi_{i}^{\prime}\rba}{\binom{\lba C_i^{\prime}\rba}{d}}\Bigg] \geq \max_{t\in[0,\gamma k]} \Bigg\{ p_i+(1-(1-t)^d)(p_1-p_i)\Bigg\}
\end{aligned}
\end{equation}
Note that $p_i$ equals $p_1$ in this case. In general, though, the estimation of all the parameters have a similar formula, and therefore we use $p_i$ still. 
Since $d$ is constant, \eqref{eq2} becomes $p_i-o(\max_{(i,j):(r_i,r_j)\in\mcal{N}_d}p_i-p_j)$ by breaking $p_i-p_{\mcal{K}_{d}}$ into pairwise difference. Similarly, \eqref{eq3} would be $p_i+o(\max_{(i,j):(r_i,r_j)\in\mcal{N}_d}p_j-p_i)$ (since $k\gamma= o(1)$ is assumed). Together,
\begin{equation}
\label{eq5}
\Bigg|\E\Bigg[\frac{\lba\xi_{i}^{\prime}\rba}{\binom{\lba C_i^{\prime}\rba}{d}}\Bigg]-\mrm{B}_{i\cdot\mbf{1}}\Bigg| \leq o\Big(\max_{(i,j):(r_i,r_j)\in\mcal{N}_d}p_i-p_j\Big)
\end{equation}
On the other hand, by the Bernstein's inequality,
\begin{align*}
\Pr\lbp\big|\lba\xi_{i}^{\prime}\rba-\E\lba\xi_{i}^{\prime}\rba\big|>t\rbp \leq 2\exp\bigg( -\frac{t^2}{2\lp\binom{n_i-\gamma n}{d}p_1+\frac{2}{3}t\rp} \bigg)
\end{align*}
Let
\begin{align*}
t^2 &= \binom{n_i-\gamma n}{d}p_1\lp C_1\gamma n \log \gamma^{-1}+(3+\delta)\log n\rp \vee \lp2C_1\gamma n \log\gamma^{-1}+2(3+\delta)\log n\rp^2 \\
 &\lesssim \lp \frac{n^{\prime}}{k^{d-1}} \sqrt{n^{d-1}p_1\gamma \log \gamma^{-1}}+\gamma n \log\gamma^{-1}\rp^2
\end{align*}
We have
\begin{equation*}
\Pr\bigg\{ \big|\lba\xi_{i}^{\prime}\rba-\E\lba\xi_{i}^{\prime}\rba\big| > C_{\delta}\Big( \frac{n^{\prime}}{k^{d-1}}\sqrt{n^{d-1}p_1\gamma \log \gamma^{-1}}+\gamma n \log\gamma^{-1} \Big)\bigg\} \leq \exp\lp-C_1\gamma n \log\gamma^{-1}\rp n^{-(3+\delta)}
\end{equation*}
Hence, with probability at least
\begin{equation*}
1-\exp\lp-C_1\gamma n \log\gamma^{-1}\rp n^{-(3+\delta)}
\end{equation*}
, we have
\begin{align*}
\Bigg| \frac{\lba\xi_{i}^{\prime}\rba}{\binom{\lba C_i^{\prime}\rba}{d}}-\E\Bigg[\frac{\lba\xi_{i}^{\prime}\rba}{\binom{\lba C_i^{\prime}\rba}{d}}\Bigg]\Bigg| \leq C_{\delta}\lp \Big(\frac{1}{n}\Big)^{d-1}\sqrt{n^{d-1}p_1\gamma \log\gamma^{-1}}+\gamma \frac{k^3}{n^2} \log\gamma^{-1}\rp
\end{align*}
Since $k\gamma \log \gamma^{-1} = O(1)$ and with the assumption $\max_{(i,j):(r_i,r_j)\in\mcal{N}_d} m_{r_i r_j}I_{p_i p_j}\rightarrow \infty$, $p_1\approx p_{\mcal{K}_{d}} $, we further have
\begin{equation}
\label{eq4}
\Bigg| \frac{\lba\xi_{i}^{\prime}\rba}{\binom{\lba C_i^{\prime}\rba}{d}}-\E\Bigg[\frac{\lba\xi_{i}^{\prime}\rba}{\binom{\lba C_i^{\prime}\rba}{d}}\Bigg] \Bigg| \leq o\Big(\max_{(i,j):(r_i,r_j)\in\mcal{N}_d}p_i-p_j\Big)
\end{equation}
at least $1-\exp\lp-C_1\gamma n \log\gamma^{-1}\rp n^{-(3+\delta)}$ in probability.
Combining \eqref{eq4}, \eqref{eq5} and apply the Union Bound over \eqref{UNION}, we have
\begin{equation*}
\Bigg|\frac{\lba\xi_{i}^{\prime}\rba}{\binom{\lba C_i^{\prime}\rba}{d}}-\mrm{B}_{i\cdot\mbf{1}}\Bigg| \leq o\Big(\max_{(i,j):(r_i,r_j)\in\mcal{N}_d}p_i-p_j\Big)
\end{equation*}
with probability at least $1-n^{-(3+\delta)}$.

The proof for the rest $\mrm{B}_{\bm{s}},\bm{s}\in[k]^d$ are all similar and thus omitted. The key observation is that by the requirement on $\gamma$, we will only have
$o(1)$ misclassification proportion. This implies that for each sample mean, the proportion of "correct" random variables will dominate the "incorrect" ones. Thus, we obtain the result of the expectation of sample mean will deviate from the true parameter no larger than $o(\max_{(i,j):(r_i,r_j)\in\mcal{N}_d}p_i-p_j)$. The second part
bound the probability that the sample mean deviates too much from its expectation. Note that we can choose a proper $t$ in the Berstein's inequality to make sure that the error probability will still be desirably small after the union bound. Hence, we complete the proof.


\section{Proof of Lemma~\ref{lma:localMLE_risk}}
\label{app:lemma4localMLE_risk}
Without loss of generality, we assume that $\pi_u$ is the identity permutation and node $u$ belongs to the first community. Also, the access index is denoted as $\bm{i}_u \eqDef (u,i_2,\ldots,i_d)$ and $M_p(t) \eqDef pe^t+1-p$ is the MGF of a $\Ber(p)$ random variable. We have
\begin{equation*}
\Pr\lbp\what{\sigma}_u(u)\neq1 \text{ and } E_u \rbp \leq \sum_{l\neq1}p_l
\end{equation*}
where $E_u$ is the event \eqref{eq:good_init} of a good initialization. On $E_u$, $p_l$ is defined as the probability of the following error event.
\begin{equation}
\label{eq:localMLE_fail}
\lbp \what{L}_u(\what{\sigma}_u,t;\mbf{A}) \geq \what{L}_u(\what{\sigma}_u,1;\mbf{A}) \rbp
\end{equation}

Recall that the initial method $\what{\sigma}_u$ determines all the assignments except for the $u$-th node before the refining process. We write $\bm{i}_u \overset{\what{\sigma}_u}{\sim} r(t_u)$ to indicate the fact that now the community relation $r$ within nodes $u,i_2,\ldots,i_d$ depends on the label of node $u$, which is to be decided. Similarly, we denote the estimated connection probability parameter $\what{p}$ as $\what{p}(t_u)$. Then, the event \eqref{eq:localMLE_fail} is equivalent to
\begin{equation*}
\bigg\{ \sum_{\bm{i}_u} \mrm{A}_{\bm{i}_u}\log\frac{\what{p}(l)(1-\what{p}(1))}{\what{p}(1)(1-\what{p}(l))} + \log\frac{1-\what{p}(1)}{1-\what{p}(l)} \geq 0 \bigg\}
\end{equation*}
Note that the summation is over all possible $i_2<\cdots<i_d$. We can also write \eqref{eq:localMLE_fail} in the form of pairwise comparison. Specifically, let $\what{\nu}_{ij} = \what{\nu}_{ij}(1,l) \eqDef \log\frac{\what{p}_i(l)(1-\what{p}_j(1))}{\what{p}_j(1)(1-\what{p}_i(l))}$ and $\what{\lambda}_{ij} = \what{\lambda}_{ij}(1,l) \eqDef \log\frac{1-\what{p}_i(1)}{1-\what{p}_j(l)}$. The error event is thus further equal to
\begin{equation}
\label{eq:error_paircomp}
\begin{aligned}
\bigg\{ \sum_{(i,j):(r_i,r_j)\in\mcal{N}_d} &\Big( \sum_{\substack{\bm{i}_u\overset{\what{\sigma}_u}{\sim}r(1)=r_i \\ \bm{i}_u\overset{\what{\sigma}_u}{\sim}r(l)=r_j}} \what{\nu}_{ij}\mrm{A}_{\bm{i}_u} + \what{\lambda}_{ij} + \sum_{\substack{\bm{i}_u\overset{\what{\sigma}_u}{\sim}r(1)=r_j \\ \bm{i}_u\overset{\what{\sigma}_u}{\sim}r(l)=r_i}} \what{\nu}_{ji}\mrm{A}_{\bm{i}_u} + \what{\lambda}_{ji} \Big) \geq 0 \bigg\}
\end{aligned}
\end{equation}
The inner two summations contain $n_{ij}^{(1,l)}$ and $n_{ji}^{(1,l)}$ random variables, respectively, where
\begin{equation*}
n_{ij}^{(1,l)} \eqDef \lba\lbp \bm{i}_u \mid \bm{i}_u\overset{\what{\sigma}_u}{\sim}r(1)=r_i \text{ and } \bm{i}_u\overset{\what{\sigma}_u}{\sim}r(l)=r_j \rbp\rba
\end{equation*}

Observe that not all $\mrm{A}_{\bm{i}_u}$'s in the summand associated with $n_{ij}^{(1,l)}$ would really be $\Ber(p_i)$. The reason is that there are still a few nodes misclassified by the initialization $\what{\sigma}_u$. Nevertheless, since we require that $\what{\sigma}_u$ satisfy Condition~\ref{con:1st_step}, it can be shown that there are only $o(1) n_{ij}^{(1,l)}$ of random variables in the summand associated with $n_{ij}^{(1,l)}$ are not $\Ber(p_i)$. Therefore, we can apply the Chernoff bound on $\Pr\lbp\eqref{eq:error_paircomp}\rbp$ to obtain
\begin{equation}
\label{ORZ}
\Pr\lbp\eqref{eq:error_paircomp}\rbp \leq \prod_{(i,j):(r_i,r_j)\in\mcal{N}_d} \lp \text{Part 1}\cdot\text{Part 2} \rp
\end{equation}
where
\begin{equation*}
\text{Part 1} = \exp\Big( -\frac{1}{2}\what{\lambda}_{ji}(n_{ji}^{(1,l)}-n_{ij}^{(1,l)}) \Big) \cdot M_{p_j}\big(\frac{\what{\nu}_{ij}}{2}\big)^{n_{ji}^{(1,l)}} M_{p_i}\big(\frac{-\what{\nu}_{ij}}{2}\big)^{n_{ij}^{(1,l)}}
\end{equation*}
and
\begin{equation*}
\text{Part 2} = \lb\sup_{p\in\{p_1,\ldots,p_{\mcal{K}_{d}}\}} \frac{M_p(\frac{\what{\nu}_{ij}}{2})}{M_j(\frac{\what{\nu}_{ij}}{2})}\rb^{O(k\gamma)n_{ji}^{(1,l)}} \cdot \lb\sup_{p\in\{p_1,\ldots,p_{\mcal{K}_{d}}\}} \frac{M_p(-\frac{\what{\nu}_{ij}}{2})}{M_i(\frac{-\what{\nu}_{ij}}{2})}\rb^{O(k\gamma)n_{ij}^{(1,l)}}
\end{equation*}
First, since the parameter space we consider is an approximately equal-size one, each community has a size $(1\pm o(1))n^{\prime}$. In addition, Condition~\ref{con:1st_step} makes sure that the community size generated from $\what{\sigma}_u$ will still lie in $(1\pm o(1))n^{\prime}$. Thus, it is easy to find that
\begin{equation*}
n_{ij}^{(1,l)}\asymp n_{ji}^{(1,l)} \asymp m_{r_i r_j} \;\forall l\neq 1
\end{equation*}
Moreover, by a similar combinatorial argument as in our proof of Lemma \ref{lma:param_estim}, we know that the proportion of wrongly added random variables is $O(k\gamma)$. That is the reason we use $O(k\gamma)n_{ij}^{(1,l)}$ for the number of wrongly added random variables.

In the following, we are going to show that Part 1 can be upper bounded by $\exp(-(1-o(1))m_{r_i r_j}I_{p_i p_j})$ and Part 2 can be upper bounded by a vanishing term with respect to Part 1. With a similar technique as in \cite{GaoMa_15}, we could immediately prove that
\begin{equation}
\label{pfpart1}
\text{Part 1} \leq \exp(-(1-o(1))m_{r_i r_j}I_{p_i p_j})
\end{equation}
For the second part, we have, for all $i<j$,
\begin{equation*}
\max\lbp \Big|e^{\frac{\what{\nu}_{ij}}{2}}-1\Big|,\Big|e^{-\frac{\what{\nu}_{ij}}{2}}-1\Big|\rbp \leq C_2\frac{p_i-p_j}{p_i}
\end{equation*}
for some constant $C_2>0$. Thus,
\begin{align*}
\sup_{p\in\{p_1,\ldots,p_{\mcal{K}_{d}}\}} \frac{M_p(\frac{\what{\nu}_{ij}}{2})}{M_j(\frac{\what{\nu}_{ij}}{2})} &= 1 + \sup_{p\in\{p_1,\ldots,p_{\mcal{K}_{d}}\}}\frac{(p-p_j)\big(e^{\frac{\what{\nu}_{ij}}{2}}-1\big)}{p_j e^{\frac{\what{\nu}_{ij}}{2}}+1-p_j} \\
 &\leq 1 + O\Big( \sup_{p\in\{p_1,\ldots,p_{\mcal{K}_{d}}\}}\frac{(p-p_j)(p_i-p_j)}{p_i} \Big) \\
 &\leq \exp\bigg( O\Big(\sup_{p\in\{p_1,\ldots,p_{\mcal{K}_{d}}\}}\frac{(p-p_j)(p_i-p_j)}{p_i} \Big)\bigg)
\end{align*}
The second term of Part 2 can be bounded similarly. Together, Part 2 is upper bounded by
\begin{equation}
\label{pfpart2}
\begin{multlined}
\exp\bigg( O(k\gamma)m_{r_i r_j}\sup_{p\in\{p_1,\ldots,p_{\mcal{K}_{d}}\}}\frac{(p-p_j)(p_i-p_j)}{p_i}\bigg)
\end{multlined}
\end{equation}

\begin{enumerate}
\item \eqref{eq:gamma_further} holds: Then, \eqref{pfpart2} is upper bounded by
\begin{equation*}
\exp\Big(o\big(\frac{1}{\log k}\big)m_{r_i r_j}I_{p_i p_j} \Big) = \exp(o(1)m_{r_i r_j}I_{p_i p_j})
\end{equation*}
\item $k$ is a constant: Then, \eqref{pfpart2} is upper bounded by
\begin{equation*}
\exp\Big( o(1)m_{r_i r_j} \max_{(i,j):(r_i,r_j)\in\mcal{N}_d} I_{r_i r_j}\Big)
\end{equation*}
Note that this term will still be absorbed to the term in the summation $\sum_{(i,j):(r_i,r_j)\in\mcal{N}_d} m_{r_i r_j}I_{p_i p_j}$ that corresponds to $\max_{(i,j):(r_i,r_j)\in\mcal{N}_d} I_{r_i r_j}$ since $k = O(1)$.
\end{enumerate}

\noindent
Combining \eqref{pfpart1} and \eqref{pfpart2} into \eqref{ORZ} in either case, we complete the proof.

\section{Proof of Lemma~\ref{SPECLEMMADHSBM}}
\label{app:lemma4spec}
First, we state the lemmas that we are going to use.

\begin{lemma}\label{Bernsum}
For independent Bernoulli random variables $X_u \sim \Ber(p_u)$ and $p = \frac{1}{n}\sum_{u\in [n]}p_u$, we have
\begin{equation*}
\Pr\bigg\{ \sum_{u\in [n]}(X_u-p_u)\geq t\bigg\} \leq \exp\bigg(t-(np+t)\log\Big(1+\frac{t}{np}\Big) \bigg),\forall t\geq 0
\end{equation*}
\end{lemma}
\noindent
This lemma is Corollary~A.1.10 in \cite{AlonSpencer_04}.

\begin{lemma}\label{SPECLEMMA104anyD}
Consider the matrix $\mbf{A}_\mrm{H}\eqDef\mcal{L}(\mbf{A})$ derived from the unnormalized graph Laplacian for a realization hypergraph $\mbf{A}$. Also, denote $\mbf{P}_\mrm{H}\eqDef\E\mcal{L}(\mbf{A})$ as its expected version for ease of notation.
Suppose $\max_{u\in[n]} \sum_{v\in[n]} (\mrm{A}_\mrm{H})_{uv} \leq \tilde{d}$ and for any $S,T \subset [n]$, one of the following statements hold with some constant $C>0$:
\begin{enumerate}
\item  $\frac{e(S,T)}{\lba S\rba\lba T\rba\frac{\tilde{d}}{n}}\leq C$
\item $ e(S,T)\log(\frac{e(S,T)}{|S||T|\frac{\tilde{d}}{n}})\leq C|T|\log\frac{n}{|T|}$
\end{enumerate}
where $e(S,T) = \sum_{u\in S}\sum_{v\in T} (\mrm{A}_\mrm{H})_{uv}$. Then, $\sum_{(u,v)\in U}x_u(\mrm{A}_\mrm{H})_{uv}y_v \leq  C^{\prime}\sqrt{\tilde{d}} $ uniformly over all unit vectors $\bm{x},\bm{y}$, where $U = \lbp(u,v)\mid\lba x_uy_v\rba\geq \frac{\sqrt{\tilde{d}}}{n}\rbp$ and $C^{\prime}>0$ is some constant.
\end{lemma}
\noindent
Note that this is the direct result to the Lemma~21 in \cite{ChinRao_15}.

\begin{lemma}\label{SPECLEMMA114anyD}
For any $\tau > C\lp n^{d-1}p_1+1\rp$ with some sufficiently large $C>0$, we have
\begin{equation*}
\big|\{u\in [n] \mid d_u \geq \tau\}\big|\leq \frac{n}{\tau}
\end{equation*}
with probability at least $1-\exp\lp-C^{\prime}n\rp$ for some constant $C^{\prime}>0$.
\end{lemma}

\begin{proof}
Note that in this lemma, the edges $e(S)$ and $e(S,S^c)$ are counting the actual hyperedges in $\mbf{A}$. This is different from the definition in Lemma~\ref{SPECLEMMA104anyD}.
\noindent
Let us consider a subset of nodes $S \subset [n]$ which contains all nodes with degree at least $\tau$ and $\lba S\rba = l$ for some $l \in [n]$. By the requirement on $S$, we have either $e(S) \geq C_1l\tau$ or $e(S,S^c) \geq C_1l\tau$ for some constant $C_1$. We want to show that both $\Pr\lbp e(S) \geq C_1l\tau\rbp$ and $\Pr\lbp e(S,S^c) \geq C_1l\tau\rbp$ are small.
First, observe that $e(S)$ consists of $\binom{l}{d}$ Bernoulli random variables and $e(S,S^c)$ consists of $\sum_{s=1}^{d-1}\binom{n-l}{s}\binom{l}{d-s}$ Bernoulli random variables. Thus, $\E e(S)\leq C_2 l^dp_1$ and $\E e(S,S^c)\leq C_2 n^{d-1}lp_1$ for some constant $C_2$. Then, when $\tau > C\lp n^{d-1}p_1+1\rp$ for some sufficiently large $C>0$, we have
\begin{align*}
\Pr\lbp e(S) \geq C_1l\tau\rbp &= \Pr\lbp e(S)- \E e(S)\geq C_1l\tau-\E e(S)\rbp & \\
 &\leq \exp\lp C_1l\tau-\E e(S)-C_1l\tau\log\lp\frac{C_1l\tau}{\E e(S)}\rp\rp &\text{by Lemma }\ref{Bernsum} \\
 &\leq \exp\lp C_1l\tau-C_1l\tau\log\lp\frac{C_1\tau}{C_2 n^{d-1}p_1}\rp\rp & \\
 &\leq \exp\lp C_1l\tau-C_1l\tau\log(C_3)\rp &\text{where }C_3 = \frac{C_1C}{C_2}\\
 &\leq \exp\lp-C_4l\tau\rp &\text{for some }C_4>0
\end{align*}
where the last inequality holds since $C_3$ is sufficiently large. Similarly, the same bound applies for
\begin{equation*}
\Pr\lbp e(S,S^c) \geq C_1l\tau\rbp
\end{equation*}
Thus, by Union Bound
\begin{equation*}
\Pr\lbp\lba\{u\in[n] \mid d_u \geq \tau \}\rba> \xi n\rbp \leq \sum_{l> \xi n}2\exp\lp l\log\lp\frac{en}{l}\rp\rp\cdot\exp\lp-C_4l\tau\rp\leq \exp(-C_5n)
\end{equation*}
where we choose $\xi = \frac{1}{\tau}$. We are done.
\end{proof}

\begin{lemma}\label{SPECLEMMA124anyD}
Given $\tau > 0$, define the subset $J = \lbp u\in [n] \mid d_u \leq \tau\rbp$. Then for any $C^{\prime}>0$, there is some constant $C>0$ such that
\begin{equation*}
\lnorm (\mbf{A}_\mrm{H})_{JJ}-(\mbf{P}_\mrm{H})_{JJ}\rnorm_{\mrm{op}} \leq C\lp\sqrt{n^{d-1}p_1}+\sqrt{\tau}+\frac{n^{d-1}p_1}{\sqrt{n^{d-1}p_1}+\sqrt{\tau}} \rp
\end{equation*}
with probability at least $1-n^{-C^{\prime}}$.
\end{lemma}

\begin{proof}
By definition,
\begin{equation*}
\lnorm (\mbf{A}_\mrm{H})_{JJ}-(\mbf{P}_\mrm{H})_{JJ}\rnorm_{\mrm{op}} = \sup_{\bm{x},\bm{y}\in S^{n-1}}\sum_{(u,v)\in J\times J}x_u\big((\mrm{A}_\mrm{H})_{uv}-(\mrm{P}_\mrm{H})_{uv}\big) y_v
\end{equation*}
where $\bm{x},\bm{y}$ are some unit vectors lying on the unit sphere $S^{n-1}$ in $\mbb{R}^{n-1}$.
Define the following two sets
\begin{gather*}
L = \lbp(u,v) : \lba x_uy_v\rba\leq \lp\sqrt{\tau} + \sqrt{n^{d-1}p_1}\rp/n\rbp \\
U = \lbp(u,v) : \lba x_uy_v\rba\geq \lp\sqrt{\tau} + \sqrt{n^{d-1}p_1}\rp/n\rbp
\end{gather*}
Then we have
\begin{align*}
\lnorm (\mbf{A}_\mrm{H})_{JJ}-(\mbf{P}_\mrm{H})_{JJ}\rnorm_{\mrm{op}} &\leq \sup_{\bm{x},\bm{y}\in S^{n-1}}\sum_{(u,v)\in L\cap J\times J}x_u\big((\mrm{A}_\mrm{H})_{uv}-(\mrm{P}_\mrm{H})_{uv}\big)y_v \\
 &+ \sup_{\bm{x},\bm{y}\in S^{n-1}}\sum_{(u,v)\in  U\cap J\times J}x_u\big((\mrm{A}_\mrm{H})_{uv}-(\mrm{P}_\mrm{H})_{uv}\big)y_v
\end{align*}
We will upper-bound these two parts separately. First we bound the light pairs $\lbp (u,v) \in L \rbp$. A discretization argument as in \cite{ChinRao_15} implies that
\begin{equation*}
\sup_{\bm{x},\bm{y}\in S^{n-1}}\sum_{(u,v)\in L\cap J\times J}x_u\big((\mrm{A}_\mrm{H})_{uv}-(\mrm{P}_\mrm{H})_{uv}\big)y_v \lesssim \max_{\bm{x},\bm{y}\in \mcal{N}}\max_{S\subset [n]}\sum_{(u,v)\in L\cap S\times S}x_u\big((\mrm{A}_\mrm{H})_{uv}-\E(\mrm{A}_\mrm{H})_{uv}\big)y_v
\end{equation*}
where $\mcal{N} \subset S^{n-1}$ and $\lba\mcal{N}\rba\leq 5^n$.
Let $r_{uv} = x_uy_v\Indc{\lba x_uy_v\rba\leq \sqrt{\tilde{d}}/n}$ and $\sqrt{\tilde{d}}=\sqrt{\tau}+\sqrt{n^{d-1}p_1}$. Then,
   \begin{align*}\label{BoundingLP}
      &\Pr\Bigg\{\Bigg|\sum_{u<v}r_{uv}\big((\mrm{A}_\mrm{H})_{uv} -\E(\mrm{A}_\mrm{H})_{uv}\big) \Bigg|\geq C\sqrt{\tilde{d}}\Bigg\} & \\
      &= \Pr\Bigg\{\Bigg|\sum_{u<v}\sum_{\bm{i}_3^{d}:i_3<\cdots<i_d}r_{uv}\big(\mrm{A}_{u,v,\bm{i}_3^{d}} -\E \mrm{A}_{u,v,\bm{i}_3^{d}}\big) \Bigg|\geq C\sqrt{\tilde{d}}\Bigg\} & \text{by definition}\\
      &= \Pr\Bigg\{\Bigg|\sum_{u<v}\sum_{\bm{i}_3^d:i_3<\cdots<i_d}r_{uv}\big(\mrm{A}_{u,v,\bm{i}_3^{d}} -\E \mrm{A}_{u,v,\bm{i}_3^{d}}\big)\Bigg|\geq C_{1}\sqrt{\tilde{d}}\Bigg\} & \text{where }C_{1} = C\times (d-2)!\\
      &= \Pr\Bigg\{\Bigg|\sum_{\bm{i}_1^d:i_1<\cdots<i_d}(\sum_{1\leq a < b \leq d}r_{i_a i_b})\big(\mrm{A}_{\bm{i}_1^{d}} -\E \mrm{A}_{\bm{i}_1^{d}}\big)\Bigg|\geq C_{1}\sqrt{\tilde{d}}\Bigg\} & \text{simple rearrangement according to independent terms}\\
      &\leq \sum_{1\leq a < b \leq d}\Pr\Bigg\{\Bigg|\sum_{\bm{i}_1^d:i_1<\cdots<i_d}(r_{i_a i_b})\big(\mrm{A}_{\bm{i}_1^{d}} -\E \mrm{A}_{\bm{i}_1^{d}}\big)\Bigg|\geq C_{2}\sqrt{\tilde{d}}\Bigg\} & (a)\text{  Union bound, }C_2 = C_1/\binom{d}{2}\\
      &\leq 2\sum_{1\leq a < b \leq d}\exp\lp -\frac{1/2C_2^2\tilde{d}}{p_1\sum_{\bm{i}_1^d:i_1<\cdots<i_d}r_{i_a i_b}^2+\frac{2}{3}\frac{\sqrt{\tilde{d}}}{n}C_2\sqrt{\tilde{d}}} \rp & \text{Bernstein's inequality}\\
      &\leq 2\binom{d}{2}\exp\lp -\frac{1/2C_2^2\tilde{d}}{2p_1n^{d-2}+\frac{2}{3}\frac{\sqrt{\tilde{d}}}{n}C_2\sqrt{\tilde{d}}} \rp & (b)\\
      &\leq 2\binom{d}{2}\exp\Big( -n\frac{C_2^2}{4+\frac{4C_2}{3}} \Big) & \text{since }\tilde{d}>p_1n^{d-1}
   \end{align*}
The inequality $(b)$ holds since $\sum_{i_a<i_b}r_{i_a i_b}^2 \leq 2\sum_{1\leq i_a<i_b \leq n}x_{i_a}^2y_{i_b}^2 \leq 2n^{d-2} $ (recall that $\bm{x},\bm{y}$ are all unit vectors).
Then, we apply the Union Bound over the space $\mcal{N}$ and the other half of the Laplacian matrix $\mbf{A}_\mrm{H}$, we have
\begin{equation*}
\max_{\bm{x},\bm{y}\in \mcal{N}}\max_{S\subset [n]}\sum_{(u,v)\in L\cap S\times S}x_u\big((\mrm{A}_\mrm{H})_{uv}-\E(\mrm{A}_\mrm{H})_{uv}\big) y_v \leq C\lp\sqrt{\tau}+\sqrt{n^{d-1}p_1}\rp
\end{equation*}
with probability at least $1-\exp\lp-C^{\prime}n\rp$.
Thus, we complete the bound for light pairs. Here we want to highlight that the the above argument are all similar to \cite{ChinRao_15}, except the key step $(a)$. Step $(a)$ allows us to obtain a similar result under the $d$-{\hSBM} setting.\\
Next we show how to bound the heavy pairs $\lbp (u,v) \in U \rbp$. Similar to \cite{GaoMa_15}, we bound
\begin{equation}
\label{eq:heavyA}
\sup_{\bm{x},\bm{y}\in S^{n-1}}\sum_{(u,v)\in  U\cap J\times J}x_u(\mrm{A}_\mrm{H})_{uv}y_v
\end{equation}
and
\begin{equation*}
\sup_{\bm{x},\bm{y}\in S^{n-1}}\sum_{(u,v)\in  U\cap J\times J}x_u(\mrm{P}_\mrm{H})_{uv}y_v
\end{equation*}
separately. By the definition of $U$, we have
\begin{equation*}
   \sup_{\bm{x},\bm{y}\in S^{n-1}}\sum_{(u,v)\in  U\cap J\times J}x_u(\mrm{P}_\mrm{H})_{uv}y_v \leq \sup_{\bm{x},\bm{y}\in S^{n-1}}\sum_{(u,v)\in  U\cap J\times J}\frac{x_u^2y_v^2}{|x_uy_v|}(\mrm{P}_\mrm{H})_{uv}\leq \frac{n^{d-1}p_1}{\sqrt{n^{d-1}p_1}+\sqrt{\tau}}
\end{equation*}
The last equation hold since $\max_{u,v}(\mrm{P}_\mrm{H})_{uv}\leq n^{d-2} p_1$. Then, we bound \eqref{eq:heavyA}. Note that by the definition of the set $J$, the degree of the sub-graph $(\mrm{A}_\mrm{H})_{JJ}$ is bounded above by $\tau$. We need to prove that the condition (the \emph{discrepancy} property) of Lemma \ref{SPECLEMMA104anyD} is satisfied with $\tilde{d} = \tau+n^{d-1}p_1$ with probability at least $1-n^{-C^{\prime}}$. The proof mainly follows the arguments in \cite{LeiRinaldo_15} and apply the Union Bound to make sure the independence (like what we have done in $(a)$ above) . We have
\begin{equation*}
\sup_{\bm{x},\bm{y}\in S^{n-1}}\sum_{(u,v)\in  U\cap J\times J}x_u(\mrm{A}_\mrm{H})_{uv}y_v \leq C\lp\sqrt{\tau}+\sqrt{n^{d-1}p_1}\rp
\end{equation*}
with probability at least $1-n^{-C^{\prime}}$. Together with all the results above, we are done.
\end{proof}

Now we are ready to prove Lemma~\ref{SPECLEMMADHSBM}.

\begin{proof}
By triangle inequality,
\begin{equation*}
\lnorm T_{\tau}(\mbf{A}_\mrm{H})-\mbf{P}_\mrm{H}\rnorm_{\mrm{op}} \leq \lnorm T_{\tau}(\mbf{A}_\mrm{H})-T_{\tau}(\mbf{P}_\mrm{H})\rnorm_{\mrm{op}}+\lnorm T_{\tau}(\mbf{P}_\mrm{H})-\mbf{P}_\mrm{H}\rnorm_{\mrm{op}}
\end{equation*}
Then we have $\lnorm T_{\tau}(\mbf{A}_\mrm{H})-T_{\tau}(\mbf{P}_\mrm{H})\rnorm_{\mrm{op}} = \lnorm (\mbf{A}_\mrm{H})_{JJ}-(\mbf{P}_\mrm{H})_{JJ}\rnorm_{\mrm{op}}$, which is bounded by Lemma \ref{SPECLEMMA124anyD}. By Lemma \ref{SPECLEMMA114anyD}, we have $\lba J^c\rba\leq \frac{n}{\tau}$ with probability at least $1-\exp\lp-C^{\prime}n\rp$. This implies
\begin{equation*}
   \lnorm T_{\tau}(\mbf{P}_\mrm{H})-\mbf{P}_\mrm{H}\rnorm_{\mrm{op}} \leq \lnorm T_{\tau}(\mbf{P}_\mrm{H})-\mbf{P}_\mrm{H}\rnorm_{\mrm{F}} \leq \sqrt{2n\lba J^c\rba \big(\max_{i,j}(\mrm{P}_\mrm{H})_{u,v}\big)^2} \leq \frac{\sqrt{2}n\cdot\max_{i,j}(\mrm{P}_\mrm{H})_{u,v}}{\sqrt{\tau}}\leq \frac{\sqrt{2}n^{d-1}p_1}{\sqrt{\tau}}
\end{equation*}
Taking $\tau \in \lb C_1(1+n^{d-1}p_1),C_2(1+n^{d-1}p_1)\rb$. Proof completed.
\end{proof} 

\section{Proof of Lemma~\ref{HOMOEIG4anyD}}
\label{app:lemma5spec}
We start from analyzing the entries of $\mbf{P}_\mrm{H}$. Recall that $\mbf{P}_\mrm{H}\eqDef\E\mcal{L}(\mbf{A})$ for an adjacency tensor $\mbf{A}$. Under the transformation from a $d$-dimensional tensor into a two-dimensional matrix, each entry in $\mbf{P}_\mrm{H}$ is a weighted combination of the probability parameters $p_i$'s. To be specific, $(\mrm{P}_{\mrm{H}})_{ij}$ aggregates the contribution from other nodes $u\in[n]\setminus\{u,v\}$, and the value depends on the community relation induced by each hyperedge correspondingly. Depending on whether or not the two nodes $u$ and $v$ are in the same community, we have, $\forall i\neq j$
\begin{equation}
    (\mrm{P}_\mrm{H})_{ij} \approx
    \begin{cases}
      u & \sigma(i) = \sigma(j) \\
      v & \mbox{, otherwise.}
    \end{cases}
\end{equation}
The explicit expression for $(\mrm{P}_{\mrm{H}})_{ij}$ changes for different values of $d$, the order of the underlying hypergraph. Observe that $u\geq v$ since we assume that $p_i$'s are in decreasing order, i.e. $p_i>p_j$ for $i<j$. Below are $u,v$ for the case $d=3$ and $4$.
\begin{equation}
    \text{When }d=3
    \begin{cases}
      u \approx n^{\prime}(p+(k-1)q) \\
      v \approx n^{\prime}(2q+(k-2)r)
    \end{cases}
\end{equation}
\begin{equation}
    \text{When }d=4
    \begin{cases}
      u \approx (n^{\prime})^2\big(\frac{1}{2} p_1 + (k-1)p_2 + \frac{k-1}{2} p_3 + \binom{k-1}{2} p_4\big) \\
      v \approx (n^{\prime})^2\big( p_2 + p_3 + \frac{5(k-2)}{2} p_4 + \binom{k-2}{2} p_5\big)
    \end{cases}
\end{equation}
Deducting $v$ for each entry in $\mbf{P}_{\mrm{H}}$, we have
\begin{equation}
   \mbf{P}_\mrm{H} - (1-\eta)v\mbf{1}_n\mbf{1}_n^T \approx (u-v)\sum_{t=1}^{k}\bm{v}_t\bm{v}_t^T
\end{equation}
where $\bm{v}_t$ is defined as $\bm{v}_t = \lp\mbf{0}_{n_1}^{T},\ldots,\mbf{1}_{n_t}^{T},\ldots,\mbf{0}_{n_k}^{T}\rp^{T}$ for each $t \in [k]$. Note that $\lbp \bm{v}_t\rbp_{t=1}^{k}$ are orthogonal to each other. Therefore,
\begin{equation*}
   \lambda_k\lp\sum_{t=1}^{k}\bm{v}_t\bm{v}_t^T\rp \geq \min_{t\in[k]}n_t \geq (1-\eta)n^{\prime}
\end{equation*}
By Weyl's inequality,
\begin{equation*}
\lambda_k(\mbf{P}_\mrm{H}) \geq \lp u-v\rp\lambda_k\lp\sum_{t=1}^{k}\bm{v}_t\bm{v}_t^T\rp + \lambda_n\lp(1-\eta)v\mbf{1}_n\mbf{1}_n^T\rp \gtrsim (n^{\prime})\lp u-v\rp
\end{equation*}
To further control $u-v$, let's first look at a few cases for lower-order $d$.
For the case $d=3$, we have
\begin{equation*}
u-v \approx n^{\prime}\big( p-q+(k-2)(q-r)\big)
\end{equation*}
while for the case $d=4$,
\begin{equation*}
u-v \approx (n^{\prime})^2\lb \frac{1}{2}(p_1-p_2)+\frac{1}{2}(p_2-p_3)+(k-2)(p_3-p_4) + \frac{k-2}{2}(p_2-p_4)+\binom{k-1}{2}(p_4-p_5) \rb
\end{equation*}
Note that $u-v$ could be represented as a weighted sum of pairwise comparisons, that is, $\sum_{i<j:(r_i,r_j)\in \mcal{N}_d} M_{r_{i}r_{j}}(p_{i}-p_{j})$ for some $M_{r_i,r_j}$'s. Recall that in our definition, $(i,j):(r_i,r_j)\in \mcal{N}_d$ if the hyperedges of type $r_i$ and $r_j$ have community assignments that differ on only one node. The new coefficient $M_{r_{i}r_{j}}$ would be similar to $m_{r_{i}r_{j}}$. In fact, they will only differ up to a constant related only to $d$ (in fact, up to $d-1$).

Moreover, $n^{\prime} C_{r_{i}r_{j}} \geq m_{r_{i}r_{j}}$ for all possible $(r_{i},r_{j})$. When counting, in $m_{r_i,r_j}$ we fix one dimension (the first dimension to node $1$), while in $M_{r_i,r_j}$ two dimensions are fixated at $u$ and $v$. Essentially, $u-v$ counts the difference of the number of random variables between two assignments, one being $\sigma(u)=\sigma(v)$ and the other being otherwise. Without loss of generality, we may think of the community labeled of $u$ as a fixed number as in the operational definition of $m_{r_i,r_j}$, while the community label of $v$ should be different from $\sigma(u)$. By multiplying back $n^{\prime}$ to get the expression $n^{\prime}M_{r_i,r_j}$, we unshackle $v$ and allow it to vary within the $\sigma(v)$-th community, the cardinality of which is approximately $n^{\prime}$. Undoubtedly, there are double countings in both the number $m_{r_i,r_j}$ and $M_{r_i,r_j}$. The value of $m_{r_i,r_j}$ is normalized with respect to $d-1$ companions (only one dimension is fixed), while the value of $M_{r_i,r_j}$ is normalized with respect to $d-2$ companions (two dimension are fixed). As a result, there are still some $\bm{l}=(u,v,l_3,\ldots,l_d)$'s being doubled counted in coefficient $n^{\prime}M_{r_i,r_j}$ as opposed to coefficient $m_{r_i,r_j}$. This is the reason why the former is always larger than or equal to the latter.

Recall that the probability parameter $\bm{p}=\{p_1,\ldots,p_{\kappa_d}\}$ follows the majorization rule, which means that $p_i > p_j$ for all $i<j$.
Combined with these fact, we have
\begin{equation*}
\lambda_k(\mbf{P}_\mrm{H}) \gtrsim (n^{\prime})\sum_{(i,j):(r_i,r_j)\in \mcal{N}_d} C_{r_{i}r_{j}}(p_{i}-p_{j}) \gtrsim \sum_{(i,j):(r_i,r_j)\in\mcal{N}_d} m_{r_i r_j} (p_i-p_j)
\end{equation*}
Hence we complete the proof. 
 
\section{Proof of Lemma~\ref{lma:hypo_test}}
\label{app:hypo_test}
First recall that
\begin{equation*}
\Risk_{\sigma \sim \Unif} (\what{\sigma}(1)) = \frac{1}{|\Theta_{d}^{L}|} \sum_{\sigma \in \Theta_{d}^{L}} \E_{\sigma} \loss(\sigma(1), \what{\sigma}(1))
\end{equation*}
In order to connect $\Risk_{\sigma \sim \Unif} (\what{\sigma}(1))$ with the risk function of a hypothesis testing problem, we shall find an equivalent form of $\E_{\sigma} \loss(\sigma(1), \what{\sigma}(1))$. The idea is to find another assignment $\sigma^{\prime}$ such that $d(\sigma, \sigma^{\prime}) = d_\mrm{H}(\sigma(1),\sigma^{\prime}(1))= 1$. $\sigma^{\prime}$ is the most indistinguishable opponent against $\sigma$ in the sense that their assignments differ by only one node. Specifically, for each $\sigma_{0} \in \Theta_{d}^{L}$, we construct a new assignment $\sigma[\sigma_{0}]$ based on $\sigma_{0}$:
\begin{equation*}
\sigma[\sigma_{0}](1) = \argmin_{2 \le i \le n} \left\{ n_{\sigma_{0}(i)} = n^{\prime} \right \}
\end{equation*}
and $\sigma[\sigma_{0}](i) = \sigma_{0}(i)$ for $2 \le i \le n$. Note that $\{ i \mid n_{\sigma_{0}(i)} = n^{\prime} \} \neq \varnothing \, \forall \sigma_{0} \in \Theta_{d}^{L}$ and $\sigma[\sigma_{0}] \in \Theta_{d}^{L}$. In addition, for any $\sigma_{1}, \, \sigma_{2} \in \Theta_{d}^{L}$, we can see that $\sigma_{1} \neq \sigma_{2}$ if and only if $\sigma[\sigma_{1}] \neq \sigma[\sigma_{2}]$. Therefore, $\{ \sigma_{0} \mid \sigma_{0} \in \Theta_{d}^{L} \} = \{ \sigma_{0} \mid \sigma_{0} \in \Theta_{d}^{L} \}$ and thus
\begin{align*}
\Risk_{\sigma \sim \Unif} (\what{\sigma}(1)) &= \frac{1}{|\Theta_{d}^{L}|} \sum_{\sigma_{0} \in \Theta_{d}^{L}} \E_{\sigma_{0}} \loss(\sigma_{0}(1), \what{\sigma}(1)) \\
 &= \frac{1}{|\Theta_{d}^{L}|} \sum_{\sigma_{0} \in \Theta_{d}^{L}} \frac{1}{2} \big( \E_{\sigma_{0}} \loss(\sigma_{0}(1), \what{\sigma}(1)) + \E_{\sigma[\sigma_{0}]} \loss(\sigma[\sigma_{0}](1), \what{\sigma}(1)) \big)
\end{align*}

In the testing problem, we can use the optimal Bayes risk as a lower bound. Let $\what{\sigma}_{\mathrm{Bayes}}$ be an assignment that achieves the minimum Bayes risk $\inf_{\what{\sigma}} \frac{1}{2} \big( \E_{\sigma_{0}} \loss(\sigma_{0}(1), \what{\sigma}(1)) + \E_{\sigma[\sigma_{0}]} \loss(\sigma[\sigma_{0}](1), \what{\sigma}(1)) \big)$. Notice that $\what{\sigma}_{\mathrm{Bayes}}(1)$ is a Bayes estimator concerning the $0$-$1$ loss, indicating that $\what{\sigma}_{\mathrm{Bayes}}(1)$ must to be the mode of the posterior distribution. Roughly speaking, the team who has a larger value of sum of the supporting random variables wins the test.

Grouping terms together according to each community relation, the log-likelihood function under the true community assignment $\sigma_0$ given an observation $\mbf{A}$ becomes
\begin{equation*}
L(\sigma_{0};\mbf{A}) = \log \Pr\{\mbf{A}\mid\sigma_0\} = \sum_{\bm{l}=(1,l_2,\ldots,l_d)}\sum_{i=1}^{\mcal{K}_d} \mrm{A}_{\bm{l}}\Indc{\bm{l}\overset{\sigma_0}{\sim}r_i} \lp \log\frac{p_i}{\bar{p_i}} + \log\bar{p_i}\rp
\end{equation*}
Similarly, we can obtain the expression $L(\sigma[\sigma_{0}];\mbf{A})$ when the underlying community assignment changes to $\sigma[\sigma_0]$. Hence, the probability of error is
\begin{align}
\label{eq:hypotest}
\begin{split}
\E_{\sigma_{0}} \loss(\sigma_{0}(1), \what{\sigma}_{\mathrm{Bayes}}(1)) &= \Pr_{\sigma_{0}} \{ L(\sigma[\sigma_{0}];\mbf{A}) \ge L(\sigma_{0};\mbf{A}) \} \\
 &= \Pr\Bigg\{ \adjustlimits\sum_{i<j:(r_i,r_j)\in\mcal{N}_d}\sum_{u=1}^{m_{r_i r_j}} \lp X_u^{(r_j)}-X_u^{(r_i)} \rp \geq 0 \Bigg\}
\end{split}
\end{align}
where where $f(s) \eqDef \frac{s}{1-s}$ for $C_{xy} \eqDef \log\frac{f(x)}{f(y)}$, and for each $(r_i,r_j)$ pair in $\mcal{N}_d$, $X_u^{(r_j)}\diid\Ber(p_j), X_u^{(r_i)}\diid\Ber(p_i) \;\forall u=1,\ldots,m_{r_i r_j}$ are all mutually indepedent random variables. Note that when summing over all possible $\bm{l}$'s in the log-likelihood function, the indices can be partitioned into two kinds of set: one whose label changes from $r_i$ to $r_j$ for some $(r_i,r_j)\in\mcal{N}_d$ when there is exactly one node disagreement and one whose label does not change whether the community assignment is $\sigma_0$ or $\sigma[\sigma_0]$. Specifically,
\begin{equation*}
\{\bm{l}=(1,l_2,\ldots,l_d)\} = J \cup J^c
\end{equation*}
where
\begin{equation*}
J = \bigcup_{i<j:(r_i,r_j)\in\mcal{N}_d} \lbp \bm{l}\overset{\sigma_0}{\sim}r_i, \bm{l}\overset{\sigma[\sigma_0]}{\sim}r_j \rbp \quad\text{ and }\quad J^c = \bigcup_{i=1}^{\mcal{K}_d} \lbp \bm{l}\overset{\sigma_0}{\sim}r_i, \bm{l}\overset{\sigma[\sigma_0]}{\sim}r_i \rbp
\end{equation*}
The former contributes to the difference between two Bernoulli random variables with cardinality $m_{r_i r_j}$, while the latter is invariant to the hypothesis testing problem and its likelihood remains the same at both sides of the first inequality in \eqref{eq:hypotest}. Note also that we rearrange terms on the specific side of the inequality to make $C_{r_i r_j}\geq 0\ \forall(r_i,r_j)\in\mcal{N}_d$ due to the non-decreasing property of the probability parameters $p_i$'s.

By symmetry, the situation is exactly the same for $\E_{\sigma[\sigma_{0}]} \loss(\sigma[\sigma_{0}](1), \what{\sigma}_{\mathrm{Bayes}}(1))$. Finally, since \eqref{eq:hypotest} holds for all $\sigma_{0} \in \Theta_{d}^{L}$ and $\inf(\cdot)$ is a concave function, we have
\begin{align*}
\Risk_{\sigma \sim \Unif} (\what{\sigma}(1)) &\ge \inf_{\what{\sigma}} \Risk_{\sigma \sim \Unif} (\what{\sigma}(1)) \\
 &= \inf_{\what{\sigma}} \frac{1}{|\Theta_{d}^{L}|} \sum_{\sigma_{0} \in \Theta_{d}^{L}} \frac{1}{2} \big( \E_{\sigma_{0}} \loss(\sigma_{0}(1), \what{\sigma}(1)) + \E_{\sigma[\sigma_{0}]} \loss(\sigma[\sigma_{0}](1), \what{\sigma}(1)) \big) \\
 &\ge \frac{1}{|\Theta_{d}^{L}|} \sum_{\sigma_{0} \in \Theta_{d}^{L}} \inf_{\what{\sigma}} \frac{1}{2} \big( \E_{\sigma_{0}} \loss(\sigma_{0}(1), \what{\sigma}(1)) + \E_{\sigma[\sigma_{0}]} \loss(\sigma[\sigma_{0}](1), \what{\sigma}(1)) \big) \\
 &= \frac{1}{|\Theta_{d}^{L}|} \sum_{\sigma_{0} \in \Theta_{d}^{L}} \Pr\Bigg\{ \adjustlimits\sum_{i<j:(r_i,r_j)\in\mcal{N}_d}\sum_{u=1}^{m_{r_i r_j}} \lp X_u^{(r_j)}-X_u^{(r_i)} \rp \geq 0 \Bigg\} \\
 &= \Pr\Bigg\{ \adjustlimits\sum_{i<j:(r_i,r_j)\in\mcal{N}_d}\sum_{u=1}^{m_{r_i r_j}} \lp X_u^{(r_j)}-X_u^{(r_i)} \rp \geq 0 \Bigg\}
\end{align*} 
 
\section{Proof of Lemma~\ref{lma:rozovsky}}
\label{app:rozovsky}
We can break the L.H.S. of \eqref{eq:lmaroz} dirctly into
\begin{equation*}
\Pr\Bigg\{ \adjustlimits\sum_{i<j:(r_i,r_j)\in\mcal{N}_d}\sum_{u=1}^{m_{r_i r_j}} \lp X_u^{(r_j)}-X_u^{(r_i)} \rp \geq 0 \Bigg\} \geq \prod_{i<j:(r_i,r_j)\in\mcal{N}_d} \Pr\Bigg\{ \sum_{u=1}^{m_{r_i r_j}} \lp X_u^{(r_j)}-X_u^{(r_i)} \rp \geq 0 \Bigg\}
\end{equation*}
Note that there are only finitely many terms involving in the product since we assume the order $d$ is constant and so does the total number of community relations $\kappa_d=|\mcal{K}_d|$ in $d$-{\hSBM}. Though na\"{i}ve, we could still arrive at the same order as the minimax rate. By symmetry, it suffices to focus on the first term in the above equation.
\begin{equation*}
\Pr\Bigg\{ \sum_{u=1}^{m_{r_1 r_2}} C_{r_1 r_2} \lp X_u^{(r_2)}-X_u^{(r_1)} \rp \geq 0 \Bigg\}
\end{equation*}
Here, we utilize a result from large deviation.

Conseder i.i.d. random variables $\{ X_{i} \}_{i=1}^{n}$ where each $X_{i} \sim X$. We assume $X$ is nondegenerate and that
\begin{equation}
\label{eq:lambda}
\mathbb{E} X^{2} e^{\lambda X} < \infty
\end{equation}
for some $\lambda > 0$. The former condition ensures, for $0 < u \le \lambda$, the existence of the functions $m(u) \eqDef \big( \log L_{X}(u) \big)^{\prime}$, $\sigma^{2}(u) \eqDef m^{\prime}(u)$ and $Q(u) \eqDef um(u) - \log L_{X}(u)$ where $L_{X}(u) \eqDef \mathbb{E} e^{uX}$ is the \emph{Moment Generating Function} (MGF) of the random variable $X$. Recall some known results:
\begin{equation*}
\lim_{u \downarrow 0} m(u) = m(0) = \mathbb{E} X < \infty
\end{equation*}
and
\begin{equation}
\label{eq:Qsup}
\sup_{0 < u \le \lambda} (ux - \log L_{X}(u)) = Q(u^{\ast})
\end{equation}
for $m(0) < x \le m(\lambda)$, where $u^{\ast}$ is the unique solution of the equation
\begin{equation}
\label{eq:ustar}
m(u) = x
\end{equation}
Note that it is the sup-achieving condition in \eqref{eq:Qsup}. The main theorem goes as follows.
\begin{theorem}[\textit{Theorem 1} in \cite{Rozovsky_03}]
$\forall x$ such that $m(0) < x \le m(\lambda)$ and $\forall n \ge 1$, the relation
\begin{equation*}
e^{-nQ(u^{\ast})} \ge \Pr \left\{ \sum_{i=1}^{n} X_{i} \ge nx \right\} \ge e^{-nQ(u^{\ast}) - c\left( 1+\sqrt{nQ(u^{\ast})} \right)}
\end{equation*}
holds, where the constant $c$ does not depend on $x$ and $n$.
\end{theorem}
\noindent
The first inequality is essentially the Chernoff Bound, while here we use the second one, i.e. the lower bound result.

First, we identify that $X = C_{r_1 r_2} ( X_u^{(r_2)}-X_u^{(r_1)} )$ and $n = m_{r_1 r_2}$ for our problem. Besides, since $X < \infty$, we can take $\lambda$ large enough so that \eqref{eq:lambda} holds. The MGF now becomes
\begin{equation*}
L_{X}(u) = \mathbb{E} e^{uX} = \mathbb{E} \big[ e^{uC_{r_1 r_2}X_u^{(r_2)}} \big] \cdot \mathbb{E} \big[ e^{-uC_{r_1 r_2}X_u^{(r_1)}} \big]
\end{equation*}
Also, since $m(0) = \mathbb{E} X < 0$, we make a trick here to take $x = 0$. The corresponding optimalilty condition \eqref{eq:ustar} becomes
\begin{align*}
m(u) = x = 0 &\Leftrightarrow \frac{L_{X}^{\prime}(u)}{L_{X}(u)} = 0 \\
 &\Leftrightarrow L_{X}^{\prime}(u) = 0
\end{align*}
It can be shown that $u^{\ast} = \frac{1}{2}$ and the supremum achieved is
\begin{align*}
Q(u^{\ast}) &= \sup_{0 < u \le \lambda} (ux - \log L_{X}(u)) \\
 &= - \log L_{X}(u^{\ast}) \\
 &= I_{p_1 p_2}
\end{align*}
Combining the expressions for each $C_{r_i r_j}$ corresponding to a $(r_i,r_j)\in\mcal{N}_d$, we can conclude that
\begin{align*}
 &\Pr\Bigg\{ \adjustlimits\sum_{i<j:(r_i,r_j)\in\mcal{N}_d}\sum_{u=1}^{m_{r_i r_j}} \lp X_u^{(r_j)}-X_u^{(r_i)} \rp \geq 0 \Bigg\} \\
 \geq\> &\prod_{i<j:(r_i,r_j)\in\mcal{N}_d} \Pr\Bigg\{ \sum_{u=1}^{m_{r_i r_j}} \lp X_u^{(r_j)}-X_u^{(r_i)} \rp \geq 0 \Bigg\} \\
 \geq\> &\prod_{i<j:(r_i,r_j)\in\mcal{N}_d} e^{-m_{r_i r_j}I_{p_i p_j} - c_{r_i r_j}\left( 1+\sqrt{m_{r_i r_j}I_{p_i p_j}} \right)} \\
 =\> &\exp \Bigg( -\sum_{i<j:(r_i,r_j)\in\mcal{N}_d}\Big( m_{r_i r_j}I_{p_i p_j} + c_{r_i r_j}\big( 1+\sqrt{m_{r_i r_j}I_{p_i p_j}} \big) \Big) \Bigg) \\
 \geq\> &\exp \Bigg( -\sum_{i<j:(r_i,r_j)\in\mcal{N}_d}\bigg( m_{r_i r_j}I_{p_i p_j} + c\Big( \kappa_d + \sqrt{\kappa_d\sum_{i<j:(r_i,r_j)\in\mcal{N}_d}m_{r_i r_j}I_{p_i p_j}} \Big) \bigg) \Bigg)
\end{align*}
where $c = \max_{i<j:(r_i,r_j)\in\mcal{N}_d} \{c_{r_i r_j}\}$ is independent of $n^{\prime}$. Finally, since we assume that $\sum_{i<j:(r_i,r_j)\in\mcal{N}_d}m_{r_i r_j}I_{p_i p_j}$ goes to infinity as $n$ becomes large, the second term with the constant $c$ in the above equation would be dominated by the first term. We have the desired asymtotic result consequently.

\end{document}